\documentclass[submission,copyright,creativecommons]{eptcs}
\providecommand{\event}{EXPRESS 2011}

\newwrite\lafont
\immediate\openout\lafont=la14.mf%
\catcode`#=11
\immediate\write\lafont{%
font_identifier "LA"; font_size 14pt#;
font_coding_scheme:="LA text";
mode_setup;
boolean serifs,monospace;
xpt#:=pt#;
monospace:=false;
serifs:=true;
pair auf,ab,an;
auf=(1,7);
ab=(-1,-7);
an=(1,0.8);
ut#:=0.3;
ut=ut#*hppp;
numeric str_d,str_dp,str_dm,str_i;
numeric top,mid,base,bot,an_top,an_mid,an_bot,ang;
numeric kl_oben,gr_oben;
str_d:=1.8ut;         
str_i:=2.8ut;
str_dp:=2.3ut;
str_dm:=1.3ut;
pen va_pen;
va_pen := pencircle scaled str_d;  
pen va_pen_small;
va_pen_small := pencircle scaled str_dm;
pen va_pen_big;
va_pen_big := pencircle scaled str_dp;
top:=8ut; mid:=-3ut; bas:=0ut; bot:=0ut;
kl_oben:=19ut; gr_oben:=34ut;
an_top=16ut; an_mid=11ut; an_bot=6ut;
beginchar(oct"043", 11ut#, 19ut#, 0ut#);    
pickup va_pen;
z1=(0ut,19ut);
z2=(11ut,19ut);
draw z1{dir -70}..tension0.8..{dir 50}z2;
labels(range 1 thru 2);
endchar;
beginchar(oct"044", 12ut#, 19ut#, 0ut#);    
pickup va_pen;
z1=(0ut,19ut);
z2=(11ut,15ut);
draw z1{dir -70}..tension0.8..{dir 50}z2;
labels(range 1 thru 2);
endchar;
beginchar(oct"045", 9ut#, 19ut#, 0ut#);    
pickup va_pen;
z1=(0ut,19ut);
z2=(9ut,14ut);
draw z1{dir -70}..tension0.8..{dir 70}z2;
labels(range 1 thru 2);
endchar;
beginchar(oct"046", 6ut#, 19ut#, 0ut#);    
pickup va_pen;
z1=(0ut,19ut);
z2=(6ut,15ut);
draw z1{dir -70}..tension0.8..{dir 50}z2;
labels(range 1 thru 2);
endchar;
beginchar(133,8ut#,19ut#,0ut);   
pickup va_pen;
z1=(bot,an_bot);
z2=(top-1ut,an_top-1ut);     
draw z1{dir 55}..{dir 50}z2;
labels(range 1 thru 2);
endchar;
beginchar(134,6ut#,19ut#,0ut);   
pickup va_pen;
z1=(bot,an_bot);
z2=(top-2ut,an_top-3ut);          
draw z1{dir 55}..{dir 50}z2;
labels(range 1 thru 2);
endchar;
beginchar(146,7ut#,19ut#,0ut);   
pickup va_pen;
z1=(bot,an_bot);
z2=(top-1ut,an_top+3ut);
draw z1{dir 55}..{dir 70}z2;
labels(range 1 thru 2);
endchar;
beginchar(132,7ut#,19ut#,0ut);   
pickup va_pen;
z1=(bot,an_bot);
z2=(top-1ut,an_top-1ut);
draw z1{dir 55}..{dir 50}z2;
labels(range 1 thru 2);
endchar;
beginchar(127,13ut#,19ut#,0ut);   
pickup va_pen;                      
z1=(bot,an_bot - 6ut);
z2=(top+4.7ut,an_top -3ut);
draw z1{dir 50}..{dir 50}z2;
labels(range 1 thru 2);
endchar;
beginchar(128,8ut#,19ut#,0ut);   
pickup va_pen;                      
z1=(bot,an_bot - 6ut);
z2=(top-0ut,an_top - 10ut);
draw z1..{dir 45}z2;
labels(range 1 thru 2);
endchar;
beginchar(129,14ut#,19ut#,0ut);   
pickup va_pen;                      
z1=(bot,an_bot - 6ut);
z2=(top+6ut,an_top+3ut);
draw z1{dir 50}..{dir 70}z2;
labels(range 1 thru 2);
endchar;
beginchar(130,14ut#,19ut#,0ut);     
pickup va_pen;                        
z1=(bot,an_bot - 6ut);
z2=(top+5ut,an_top - 1ut);
draw z1{dir 45}..{dir 50}z2;
labels(range 1 thru 2);
endchar;
beginchar(154,6ut#,19ut#,0ut);   %
pickup va_pen;                      
z1=(bot,an_bot);
z2=(top-0ut,an_top - 10ut);
draw z1{right}..{dir 45}z2;
labels(range 1 thru 2);
endchar;
beginchar(131,8ut#,19ut#,0ut);     
pickup va_pen;
z1=(bot,an_bot);
z2=(top-1ut,an_top-1ut);
draw z1{dir 55}..{dir 50}z2;
labels(range 1 thru 2);
endchar;
beginchar(143,4ut#,19ut#,0ut);     
pickup va_pen;
z1=(bot    ,an_bot);
z2=(top-4ut,an_top-3ut);
draw z1{dir 55}..{dir 50}z2;
labels(range 1 thru 2);
endchar;
beginchar(144,8ut#,19ut#,0ut);     
pickup va_pen;
z1=(bot,an_bot);
z2=(top-0ut,an_top+3ut);
draw z1{dir 55}..{dir 70}z2;
labels(range 1 thru 2);
endchar;
beginchar(145,7ut#,19ut#,0ut);     
pickup va_pen;
z1=(bot,an_bot);
z2=(top-1ut,an_top-1ut);
draw z1{dir 55}..{dir 50}z2;
labels(range 1 thru 2);
endchar;
beginchar(153,2ut#,19ut#,0ut);   %
pickup va_pen;                      
z1=(bot-12ut,an_bot-6ut);
z2=(top-6ut,an_top - 10ut);
draw z1{right}..{dir 45}z2;
labels(range 1 thru 2);
endchar;
beginchar(139,1ut#,19ut#,0ut);     
pickup va_pen;
z1=(bot-20ut,an_bot-6ut);
z2=(top-8ut,an_top-1ut);
draw z1{right}..{dir 50}z2;
labels(range 1 thru 2);
endchar;
beginchar(140,1ut#,19ut#,0ut);     
pickup va_pen;
z1=(bot-20ut,an_bot-6ut);
z2=(top-6ut,an_top-2ut);
draw z1{right}..{dir 50}z2;
labels(range 1 thru 2);
endchar;
beginchar(141,2ut#,19ut#,0ut);     
pickup va_pen;
z1=(bot-20ut,an_bot-6ut);
z2=(top-6ut,an_top+3ut);
draw z1{right}..{dir 70}z2;
labels(range 1 thru 2);
endchar;
beginchar(142,2ut#,19ut#,0ut);     
pickup va_pen;
z1=(bot-18ut,an_bot-6ut);
z2=(top-6ut,an_top-1ut);
draw z1{right}..{dir 50}z2;
labels(range 1 thru 2);
endchar;
beginchar(148,6ut#,19ut#,0ut);   
pickup va_pen;                      
z1=(bot-5ut,an_bot - 6ut);
z2=(top-2ut,an_top - 10ut);
draw z1{dir 10}..{dir 45}z2;
labels(range 1 thru 2);
endchar;
beginchar(149,11ut#,19ut#,0ut);     
pickup va_pen;                        
z1=(bot-5ut,an_bot - 6ut);
z2=(top+1ut,an_top-3ut);
draw z1{dir 20}..{dir 50}z2;
labels(range 1 thru 2);
endchar;
beginchar(150,9ut#,19ut#,0ut);   
pickup va_pen;                      
z1=(bot-5ut,an_bot - 6ut);
z2=(top+0.7ut,an_top -3ut);
draw z1{dir 20}..{dir 50}z2;
labels(range 1 thru 2);
endchar;
beginchar(151,10ut#,19ut#,0ut);   
pickup va_pen;                      
z1=(bot-5ut,an_bot - 6ut);
z2=(top+2ut,an_top+3ut);
draw z1{dir 20}..{dir 70}z2;
labels(range 1 thru 2);
endchar;
beginchar(152,8ut#,19ut#,0ut);   
pickup va_pen;
z1=(bot-5ut,an_bot - 6ut);
z2=(top+0ut,an_top - 1ut);
draw z1{dir 20}..{dir 50}z2;
labels(range 1 thru 2);
endchar;
beginchar(147,0ut#,19ut#,0ut);   
pickup va_pen;
z1=(bot-20ut,an_bot-6ut);
z2=(top-8ut,an_top-10ut);
draw z1{right}..{dir 45}z2;
labels(range 1 thru 2);
endchar;
beginchar(135,1ut#,19ut#,0ut);     
pickup va_pen;
z1=(bot-23ut,an_bot-6ut);
z2=(top-8ut,an_top-1ut);
draw z1{right}..{dir 50}z2;
labels(range 1 thru 3);
endchar;
beginchar(136,4ut#,19ut#,0ut);     
pickup va_pen;
z1=(bot-23ut,an_bot-6ut);
z2=(top-4ut,an_top+3ut);
draw z1{right}..{dir 70}z2;
labels(range 1 thru 2);
endchar;
beginchar(137,1ut#,19ut#,0ut);     
pickup va_pen;
z1=(bot-20ut,an_bot-6ut);
z2=(top-7ut,an_top-1ut);
draw z1{right}..{dir 50}z2;
labels(range 1 thru 2);
endchar;
beginchar(138,1ut#,19ut#,0ut);     
pickup va_pen;
z1=(bot-23ut,an_bot-6ut);
z2=(top-6ut,an_top-2ut);
draw z1{right}..{dir 50}z2;
labels(range 1 thru 2);
endchar;
beginchar("1",14ut#,26ut#,0);
pickup va_pen;
x18=0 ut; y18=19ut;
x24=11ut; y24=26ut;
x28=5 ut; y28=0 ut;
draw z18--z24
  & z24{ab}--z28;
penlabels(18,24,28);
endchar;
beginchar("2",20ut#,26ut#,0);
pickup va_pen;
x22=17ut*.52; y22=0*.52;
x23=29ut*.52; y23=7 ut*.52;
x24=22ut*.52; y24=50ut*.52;
x28=6ut*.52; y28=7ut*.52;
x29=30ut*.52; y29=46ut*.52;
x34=6 ut*.52; y34=40ut*.52;
x38=22ut*.52; y38=22ut*.52;
x281=0ut*.52; y281=0ut*.52;
draw z34..z24{right}..z29..z38..z28{dir 225}..z281
 & z281{z28-z281}..z22{dir 330}..{an}z23;
penlabels(22,23,24,28,29,34,38,281);
endchar;
beginchar("3",20ut#,50ut#,0);
pickup va_pen;
x1=0ut*.52;   y1=4ut*.52;
x17=24ut*.52; y17=10ut*.52;
x19=12ut*.52; y19=27ut*.52;
x22=13ut*.52; y22=0 ut*.52;
x24=22ut*.52; y24=50ut*.52;
x34=6 ut*.52; y34=40ut*.52;
x38=22ut*.52; y38=22ut*.52;
draw z34..z24{right}..z19{left}
 & z19{right}..z38..z17..z22{left}..z1;
penlabels(1,17,19,22,24,34,38);
endchar;
beginchar("4",18ut#,50ut#,0);
pickup va_pen;
x4=27ut*.52;  y4=16ut*.52;
x22=13ut*.52; y22=0 ut*.52;
x24=15ut*.52; y24=50ut*.52;
x30=19ut*.52; y30=27ut*.52;
x32=16ut*.52; y32=-22ut*.52;
x56=-3ut*.52; y56=16ut*.52;
draw z24{ab}--z56
 & z56{dir 10}--z4;
draw z30--z22;
penlabels(4,22,24,30,32,56);
endchar;
beginchar("5",24ut#,50ut#,0);
pickup va_pen;
x2=8ut*.52;   y2=0ut*.52;
x31=0 ut*.52; y31=7 ut*.52;
x35=12ut*.52; y35=50ut*.52;
x38=28ut*.52; y38=22ut*.52;
x58=35ut*.52; y58=50ut*.52;
x66=2 ut*.52; y66=25ut*.52;
draw z35--z58;
draw z35--z66
 & z66{dir 20}..z38{down}..z2{left}..z31;
penlabels(2,31,35,38,58,66);
endchar;
beginchar("6",17ut#,50ut#,0);
pickup va_pen;
x23=21ut*.52; y23=8ut*.52;
x24=22ut*.52; y24=50ut*.52;
x30=15ut*.52; y30=27ut*.52;
x56=-5ut*.52; y56=15ut*.52;
x57=8ut*.52; y57=0 ut*.52;
draw z56..z30..z23..z57{left}..z56{up}..z24{an};
penlabels(23,24,30,56,57);
endchar;
beginchar("7",21ut#,50ut#,0);
pickup va_pen;
x19=12ut*.52; y19=27ut*.52;
x20=30ut*.52; y20=27ut*.52;
x28=6 ut*.52; y28=0 ut*.52;
x36=25ut*.52; y36=46ut*.52;
x37=32ut*.52; y37=50ut*.52;
x60=3 ut*.52; y60=43ut*.52;
x61=11 ut*.52; y61=50ut*.52;
draw z60..z61{right}..z36{right}..z37
 & z37--z28;
draw z19--z20;
penlabels(19,20,28,36,37,60,61);
endchar;
beginchar("8",22ut#,50ut#,0);
pickup va_pen;
x19=12ut*.52; y19=27ut*.52;
x23=25ut*.52; y23=8ut*.52;
x31=0 ut*.52; y31=9 ut*.52;
x33=11ut*.52; y33=-18ut*.52;
x34=6 ut*.52; y34=40ut*.52;
x38=22ut*.52; y38=22ut*.52;
x57=10ut*.52; y57=0 ut*.52;
x67=27ut*.52; y67=42ut*.52;
x68=18ut*.52; y68=50ut*.52;
x1=27ut*.52; y1=50ut*.52;
x2=31ut*.52; y2=42ut*.52;
x3=39ut*.52; y3=50ut*.52;
draw z31{up}..z19..z67{up}..z68{left}..z34{down}..z38..z23..z57{left}..z31{up};
draw z1{dir -70}..z2{right}..z3;
penlabels(19,23,31,33,34,38,57,67,68,1,2,3);
endchar;
beginchar("9",20ut#,50ut#,0);
pickup va_pen;
x1=0*.52;     y1=4ut*.52;
x3=29ut*.52;  y3=50ut*.52;
x11=32ut*.52; y11=50ut*.52;
x34=6 ut*.52; y34=40ut*.52;
x37=29ut*.52; y37=40ut*.52;
x38=25ut*.52; y38=29ut*.52;
x57=10ut*.52; y57=0 ut*.52;
x68=18ut*.52; y68=50ut*.52;
x69=15ut*.52; y69=22ut*.52;
x70=26ut*.52; y70=40ut*.52;
draw z68{left}..z34..z69{right}..z70..z68{left};
draw z3{ab}..z38{dir 255}..z57{left}..z1;
penlabels(1,3,11,34,37,38,57,68,69,70);
endchar;
beginchar("0",24ut#,50ut#,0);
pickup va_pen;
x1=30ut*.52; y1=50ut*.52;
x2=34ut*.52; y2=42ut*.52;
x3=42ut*.52; y3=50ut*.52;
x17=23ut*.52; y17=10ut*.52;
x20=30ut*.52; y20=41ut*.52;
x24=23ut*.52; y24=50ut*.52;
x46=28ut*.52; y46=22ut*.52;
x56=0ut*.52; y56=15ut*.52;
x57=11ut*.52; y57=0 ut*.52;
x58=36ut*.52; y58=50ut*.52;
draw z20..z24{left}..z56{down}..z57{right}..z17..z46..cycle;
draw z1{dir -70}..z2{right}..z3;
penlabels(1,2,3,17,20,24,46,56,57,58);
endchar;
beginchar("+",30ut#,20ut#, 0ut#);
pickup va_pen;
z1=(10ut,15ut);
z2=(25ut,15ut);
z3=(17ut,7ut);
z4=(18ut,23ut);
draw z1..z2; draw z3..z4;
labels(range 1 thru 4);
endchar;
beginchar("-",30ut#,20ut#, 0ut#);
pickup va_pen;
z1=(10ut,15ut);
z2=(25ut,15ut);
draw z1..z2;
labels(range 1 thru 2);
endchar;
beginchar(oct"025",35ut#,20 ut#, 0ut#);  
pickup va_pen_small;
z1=(10ut,15ut);
z2=(30ut,15ut);
draw z1..z2;
labels(range 1 thru 2);
endchar;
beginchar(oct"026",40ut#,20 ut#, 0ut#);   
pickup va_pen_small;
z1=(10ut,15ut);
z2=(35ut,15ut);
draw z1..z2;
labels(range 1 thru 2);
endchar;
beginchar("*",30 ut#,20 ut#, 0ut#);
pickup va_pen;
z9=(15ut,15ut);
z1=z9+(0ut,-8ut);
z2=z9+(8ut,0ut);
z3=z9+(0ut,8ut);
z4=z9+(-8ut,0ut);
z5=z9+(-5ut,-5ut);
z6=z9+(5ut,-5ut);
z7=z9+(5ut,5ut);
z8=z9+(-5ut,5ut);
draw z1..z3; draw z2..z4; draw z5..z7; draw z6..z8;
labels(range 1 thru 9);
endchar;
beginchar(":",15 ut#,20 ut#, 0ut#);
pickup pencircle scaled  str_i;
z1=(10ut,15ut);
z2=(9ut,5ut);
drawdot z1; drawdot z2;
labels(range 1 thru 2);
endchar;
beginchar("<",30 ut#,20 ut#,0 ut#);
pickup va_pen;
z1=(28ut,23ut);
z2=(10ut,15ut);
z3=(26ut,7ut);
draw z1..z2 & z2..z3;
labels(range 1 thru 3);
endchar;
beginchar(">", 30ut#, 20ut#, 0ut#);
pickup va_pen;
z1=(13ut,23ut);
z2=(30ut,15ut);
z3=(12ut,7ut);
draw z1..z2 & z2..z3;
labels(range 1 thru 3);
endchar;
beginchar("=", 30ut#,20 ut#, 0ut#);
pickup va_pen;
z1=(10ut,15ut);
z2=(25ut,15ut);
z3=(x1,y1-8ut);
z4=(x2,y2-8ut);
draw z1..z2  ; draw z3..z4;
labels(range 1 thru 4);
endchar;
beginchar("(", 13ut#, 34ut#, 0ut#);
pickup va_pen;
z1=(11ut,34ut);
z2=(3ut,-2ut);
draw z1{dir -115}..{dir -80}z2;
labels(range 1 thru 2);
endchar;
beginchar(")",13 ut#,34 ut#, 0ut#);
pickup va_pen;
z1=(11ut,34ut);
z2=(3ut,-2ut);
draw z1{dir -80}..{dir -115}z2;
labels(range 1 thru 2);
endchar;
beginchar("/",15 ut#,34 ut#, 0ut#);
pickup va_pen;
z1=(12ut,34ut);
z2=(3ut,-2ut);
draw z1..z2;
labels(range 1 thru 2);
endchar;
beginchar(oct"022",18 ut#,20 ut#, 7ut#);   
pickup va_pen;
z1=(5ut,3ut);
z2=(3ut,-7ut);
z3=(x1+8ut,y1);
z4=(x2+8ut,y2);
draw z1{dir -90}..z2;
draw z3{dir -90}..z4;
labels(range 1 thru 4);
endchar;
beginchar(oct"021", 18ut#, 36ut#, ut#);   
pickup va_pen;
z1=(5ut,36ut);
z2=(3ut,26ut);
z3=(x1+8ut,y1);
z4=(x2+8ut,y2);
draw z1{dir -90}..z2;  draw z3{dir -90}..z4;
labels(range 1 thru 4);
endchar;
beginchar(oct"20", 18ut#, 36ut#, ut#); 
pickup va_pen;
z1=(5ut,36ut);
z2=(3ut,26ut);
z3=(x1+8ut,y1);
z4=(x2+8ut,y2);
draw z1{dir -90}..z2;  draw z3{dir -90}..z4;
labels(range 1 thru 4);
endchar;
beginchar(oct"042", 18ut#, 36ut#, ut#); 
pickup va_pen;
z1=(5ut,36ut);
z2=(7ut,26ut);
z3=(x1+8ut,y1);
z4=(x2+8ut,y2);
draw z1{dir -90}..z2;  draw z3{dir -90}..z4;
labels(range 1 thru 4);
endchar;
beginchar(oct"140", 7ut#, 36ut#, ut#); 
pickup va_pen;
z1=(5ut,36ut);
z2=(7ut,26ut);
draw z1{dir -90}..z2;
labels(range 1 thru 2);
endchar;
beginchar(oct"000", 7ut#, 36ut#, ut#); 
pickup va_pen;
z1=(5ut,36ut);
z2=(7ut,26ut);
draw z1{dir -90}..z2;
labels(range 1 thru 2);
endchar;
beginchar(oct"004",18ut#,15ut#,0ut#);   
pickup va_pen_big;
z1=(-3ut,22ut);
z2=z1+(9ut,0ut);
z3=z1+(-1ut,-1ut);
z4=z2+(-1ut,-1ut);
draw z1..z3; draw z2..z4;
labels(range 1 thru 4);
endchar;
beginchar("?", 26ut#,34 ut#,0 ut#);
pickup va_pen;
z1=(10ut,30ut);
z2=(18ut,34ut);
z3=(23ut,30ut);
z4=(15ut,20ut);
z5=(8ut,10ut);
z6=(14ut,6ut);
z7=(21ut,10ut);
z8=(15ut,1ut);
draw z1..z2{right}..z3..tension1.4..z4..tension1.4..z5{down}..z6{right}..z7;
pickup pencircle scaled  str_i;
drawdot z8;
labels(range 1 thru 8);
endchar;
beginchar("!",16 ut#,34 ut#, 0ut#);
pickup va_pen;
z1=(16ut,32ut);
z2=(9ut,9ut);
z3=(8ut,2ut);
draw z1..z2;
pickup pencircle scaled  str_i;
drawdot z3;
labels(range 1 thru 3);
endchar;
beginchar(";",7 ut#,20 ut#,10 ut#);
pickup va_pen;
z1=(6ut,5ut);
z2=(6ut,1ut);
z3=(2ut,-7ut);
draw z2{dir -90}..z3{dir -110};
pickup pencircle scaled  str_i;
drawdot z1;
labels(range 1 thru 3);
endchar;
beginchar(".",6 ut#,10 ut#,0 ut#);
pickup pencircle  scaled str_i;
z1=(5ut,1ut);
drawdot z1;
labels(range 1 thru 2);
endchar;
beginchar(",", 7ut#, 10ut#,7 ut#);
pickup va_pen;
z1=(6ut,3ut);
z2=(2ut,-7ut);
drawdot z1;
draw z1+(1ut,0ut){dir -80}..z2{dir -125};
labels(range 1 thru 2);
endchar;
beginchar(oct"47", 7ut#, 40ut#,0 ut#);  
pickup va_pen;
z1=(6ut,35ut);
z2=(2ut,27ut);
drawdot z1;
draw z1+(1ut,0ut){dir -80}..z2{dir -125};
labels(range 1 thru 8);
endchar;
beginchar(oct"001", 7ut#, 40ut#,0 ut#);  
pickup va_pen;
z1=(6ut,35ut);
z2=(2ut,27ut);
drawdot z1;
draw z1+(1ut,0ut){dir -80}..z2{dir -125};
labels(range 1 thru 8);
endchar;
beginchar("a",18 ut#,19 ut#,0);       
pickup va_pen;
z1=(6ut,19ut);
z2=(-4.5ut,8ut);
z3=(0ut,0ut);
z4=(10ut,13ut);
z5=(13ut,19ut);
z6=(8ut,5ut);
z7=(10ut,0ut);
z8=(18ut,6ut);
draw z1{dir 190}..{dir 250}z2{dir 250}..z3{right}..{dir 70}z4{dir 70}..cycle;
draw z5--z6{z6-z5}..z7{dir 10}..{dir 55}z8;
labels(range 1 thru 8);
endchar;
beginchar(oct"344",18 ut#,30ut#,0);       
pickup va_pen;
z1=(6ut,19ut);
z2=(-4.5ut,8ut);
z3=(0ut,0ut);
z4=(10ut,13ut);
z5=(13ut,19ut);
z6=(8ut,5ut);
z7=(10ut,0ut);
z8=(18ut,6ut);
z9=(4ut,24ut);
z10=(9ut,24ut);
z11=(6.0ut,30ut);
z12=(11.0ut,30ut);
draw z1{dir 190}..{dir 250}z2{dir 250}..z3{right}..{dir 70}z4{dir 60}..cycle;
draw z5--z6{z6-z5}..z7{dir 10}..{dir 55}z8;
draw z9--z11;
draw z10--z12;
labels(range 1 thru 12);
endchar;
beginchar("b",12 ut#,34 ut#,0);
pickup va_pen;
z1=(0ut,13ut);
z2=(12ut,31ut);
z3=(9ut,34ut);
z4=(4ut,0ut);
z5=(12ut,19ut);
z6=(25ut,15ut);
z7=(3ut,23ut);
draw z1{dir 50}..z2{up}..z3{dir 220}..z7..{dir 260}z1{dir 260}
     ..z4{dir 10}..z5{dir 110};
labels(range 1 thru 7);
endchar;
beginchar("c",10 ut#,19 ut#,0);
pickup va_pen;
z1=(10ut,18ut);
z2=(6ut,19ut);
z3=(-4.5ut,8ut);
z4=(2ut,0ut);
z5=(10ut,6ut);
draw z1..z2{dir 190}..{dir 250}z3{dir 250}..z4{dir 10}..{dir 55}z5;
labels(range 1 thru 5);
endchar;
beginchar("d",18 ut#,34 ut#,0);
pickup va_pen;
z1=(6ut,19ut);
z2=(-4.5ut,8ut);
z3=(0ut,0ut);
z4=(11ut,13ut);
z5=(18ut,34ut);
z6=(9ut,5ut);
z7=(10ut,0ut);
z8=(25ut,6ut);
z9=(18ut,6ut);
draw z1{dir 190}..{dir 250}z2{dir 250}..z3{right}..z4{dir 70}..cycle;
draw z5--z6{z6-z5}..z7{dir 10}..z9{dir 55};
labels(range 1 thru 9);
endchar;
beginchar("e",13 ut#,19 ut#,0);
pickup va_pen;
z1=(6ut,19ut);
z2=(4ut,0ut);
z3=(13ut,6ut);
z4=(0ut,6ut);
z5=(7ut,12ut);
z6=(0ut,7ut);
draw z4{dir 45}..z5
     ..z1{dir 210}..{dir 250}z6{dir 250}..z2{dir 10}..z3{dir 55};
labels(range 1 thru 6);
endchar;
beginchar(oct"037",21 ut#,19 ut#,0);  
pickup va_pen;
z1=(15ut,19ut);
z2=(7ut,12ut);
z3=(12ut,0ut);
z4=(21ut,6ut);
z5=(14ut,10ut);
z6=(18ut,14ut);
z7=(8ut,8ut);
z8=(0ut,19ut);
draw z8{dir -70}..z2..z5..z6..z1{left}
     ..{dir 260}z7{dir 260}..z3{right}..z4{dir 55};
labels(range 1 thru 8);
endchar;
beginchar("f",9 ut#,34 ut#,15ut#);
pickup va_pen;
z1=(0ut,13ut);
z2=(12ut,32ut);
z3=(9ut,34ut);
z4=(4ut,26ut);
z5=(-7ut,-15ut);
z6=(-4ut,5ut);
z7=(3ut,2ut);
z8=(9ut,6ut);
draw z1{dir 50}..z2{up}..z3{dir 220}...z4{z5-z4}--z5;
draw z6..z7...z8{dir 55};
labels(range 1 thru 8);
endchar;
beginchar("g",17 ut#,19 ut#,15ut#);
pickup va_pen;
z1=(11ut,13ut);
z2=(0ut,0ut);
z3=(-4.5ut,8ut);
z4=(6ut,19ut);
z10=(13ut,19ut);
z11=(5ut,-6ut);
z12=(0ut,-15ut);
z13=(-3ut,-11ut);
z14=(1ut,-5ut);
z15=(17ut,6ut);
draw z1{dir 250}..z2{left}..{dir 70}z3{dir 70}..{dir 10}z4..cycle;
draw z10--z11{z11-z10}..z12{left}..z13{up}..z14...z15{dir 55};  
labels(range 1 thru 15);
endchar;
beginchar("h",20 ut#,34 ut#,0);
pickup va_pen;
z1=(0ut,13ut);
z2=(12ut,31ut);
z3=(9ut,34ut);
z4=(3ut,23ut);
z5=(-4ut,0ut);
z6=(11ut,17ut);
z7=(13ut,15ut);
z8=(10ut,3ut);
z9=(12ut,0ut);
z10=(20ut,6ut);
draw z1{dir 50}..z2{up}..z3{dir 220}..z4..z1--z5
  & z5{dir 70}..tension1.3..z6{right}..z7{down}..z8{z8-z7}..z9{dir 10}..z10{dir 55};
labels(range 1 thru 10);
endchar;
beginchar("i",5 ut#,19 ut#,0);
pickup va_pen;
z1=(0ut,19ut);
z2=(-5ut,4ut);
z3=(-3ut,0ut);
z4=(5ut,6ut);
z5=(2.5ut,28ut);
draw z1--z2{z2-z1}..z3{dir 10}..z4{dir 55};
pickup pencircle scaled  str_i;
drawdot z5;
labels(range 1 thru 5);
endchar;
beginchar("j",7 ut#,19 ut#,15ut#);
pickup va_pen;
z1=(2.5ut,28ut);
z2=(0ut,19ut);
z3=(-4ut,0ut);
z4=(-12ut,-15ut);
z5=(7ut,6ut);
draw z2--z3{z3-z2}...z4{left}..z3{z5-z3}...{dir 55}z5;
pickup pencircle scaled  str_i;
drawdot z1;
labels(range 1 thru 5);
endchar;
beginchar("k", 19 ut#, 34ut#,0);
pickup va_pen;
z1=(0ut,13ut);
z2=(12ut,31ut);
z3=(9ut,34ut);
z4=(3ut,23ut);
z5=(-4ut,0ut);
z6=(11ut,17ut);
z7=(13ut,15ut);
z8=(3ut,9ut);
z9=(11ut,0ut);
z10=(19ut,6ut);
z11=(16ut,2ut);
draw z1{dir 50}..z2{up}..z3{dir 220}...z4{z5-z4}--z5
  & z5{dir 70}..tension1.3..z6{right}..z7{down}..z8
  & z8{right}..z9{dir 10}..z10{dir 55};
labels(range 1 thru 11);
endchar;
beginchar("l",9 ut#,34 ut#,0);
pickup va_pen;
z1=(0ut,13ut);
z2=(12ut,31ut);
z3=(9ut,34ut);
z4=(1ut,0ut);
z5=(9ut,6ut);
z6=(7ut,2ut);
z7=(3ut,23ut);
draw z1{dir 50}..z2{up}..z3{dir 220}..z7{z1-z7}..z1{z1-z7}    
     ..{right}z4{dir 10}..z5{dir 55};
labels(range 1 thru 7);
endchar;
beginchar("m",36 ut#,19 ut#,0);
pickup va_pen;
z1=(5ut,15ut);
z2=(0ut,0ut);
z3=(16.5ut,19ut);
z4=(18.5ut,16ut);
z5=(13.5ut,0ut);
z6=(29ut,19ut);
z7=(31ut,16ut);
z8=(27ut,3ut);
z9=(28ut,0ut);
z10=(36ut,6ut);
z11=(0ut,15ut);
z12=(4ut,19ut);
draw z11{dir 50}..{curl4}z12{curl6}..{z2-z1}z1--z2;
draw z2{z1-z2}..tension1.3..{curl2}z3{curl4}..{z5-z4}z4--z5;
draw z5{z4-z5}..tension1.3..{curl2}z6{curl4}..{z8-z7}z7
      --z8{z8-z7}..z9{dir 10}..z10{dir 55};
labels(range 1 thru 12);
endchar;
beginchar("n",24 ut#,19 ut#,0);
pickup va_pen;
z1=(5ut,15ut);
z2=(0ut,0ut);
z3=(17ut,19ut);
z4=(19.3ut,16ut);
z5=(15ut,3ut);
z6=(16ut,0ut);
z7=(24ut,6ut);
z8=(0ut,15ut);
z9=(4ut,19ut);
draw z8{dir 50}..{curl4}z9{curl6}..{z2-z1}z1--z2;
draw z2{z1-z2}..tension1.3..{curl2}z3{curl4}
     ..{z5-z4}z4--z5{z5-z4}..z6{dir 10}..z7{dir 55};
labels(range 1 thru 9);
endchar;
beginchar("o",11 ut#,19 ut#,0);
pickup va_pen;
z1=(11ut,16ut);
z2=(0ut,0ut);
z3=(-4.5ut,8ut);
z4=(6ut,19ut);
draw z1{dir -75}..z2{left}..{dir 70}z3{dir 70}
     ..{dir 10}z4..cycle;
labels(range 1 thru 4);
endchar;
beginchar(oct"366",11 ut#,30 ut#,0);       
pickup va_pen;
z1=(11ut,16ut);
z2=(0ut,0ut);
z3=(-4.5ut,8ut);
z4=(6ut,19ut);
z5=(6ut,24ut);
z6=(11ut,24ut);
z7=(8.0ut,30ut);
z8=(13.0ut,30ut);
draw z1{dir -75}..z2{left}..{dir 70}z3{dir 70}
     ..{dir 10}z4..cycle;
draw z5--z7;
draw z6--z8;
labels(range 1 thru 8);
endchar;
beginchar("p",18 ut#,19 ut#,15ut#);
pickup va_pen;
z1=(0ut,19ut);
z2=(-10ut,-15ut);
z3=(-4.5ut,4ut);
z4=(11.5ut,19ut);
z5=(13.0ut,15ut);
z6=(9ut,4ut);
z7=(10ut,0ut);
z8=(17ut,1ut);
z9=(18ut,6ut);
z10=(6ut,16ut);
draw z1--z2;
draw z3{dir 50}..z10..z4{right}..z5{z6-z5}--z6{z6-z5}..z7{dir 10}..z9{dir 55};
labels(range 1 thru 10);
endchar;
beginchar("q",23 ut#,19 ut#,15ut#);
pickup va_pen;
z1=(9ut,13ut);
z2=(0ut,0ut);
z3=(-4.5ut,8ut);
z4=(6ut,19ut);
z5=(11ut,19ut);
z6=(1ut,-15ut);
z7=(6ut,0ut);
z8=(23ut,19ut);
draw z1{dir 250}..z2{left}..{dir 70}z3{dir 70}..{dir 10}z4..cycle;
draw z5--z6;
draw z7{dir 40}..z8{dir 70};
labels(range 1 thru 8);
endchar;
beginchar("r", 10ut#,19 ut#,0);  
pickup va_pen;
z1=(5ut,15ut);
z2=(4ut,11ut);    
z3=(10ut,19ut);
z4=(0ut,0ut);
z5=(0ut,15ut);
z6=(4ut,19ut);
draw z5{dir 50}..{curl4}z6{curl6}..{z4-z1}z1--z4;
draw z2..z3;
labels(range 1 thru 6);
endchar;
beginchar("s",4 ut#,19 ut#,0);
pickup va_pen;
z1=(-10ut,7ut);
z2=(0ut,19ut);
z3=(3ut,6ut);
z4=(-2ut,0ut);
z5=(-6ut,2ut);
draw z2{dir 260}..tension0.8..z3{down}..z4{left}..z5;
labels(range 1 thru 5);
endchar;
beginchar("t",7 ut#,34 ut#,0);
pickup va_pen;
z1=(4ut,27ut);
z2=(-3.5ut,3ut);
z3=(-1ut,0ut);
z4=(7ut,6ut);
z5=(-4ut,18ut);
z6=(7ut,18ut);
draw z1..z2{z2-z1}..z3{dir 10}..z4{dir 55};
draw z5--z6;
labels(range 1 thru 6);
endchar;
beginchar("u",19 ut#,19 ut#,0);
pickup va_pen;
z1=(0ut,19ut);
z2=(-5ut,3ut);
z3=(0ut,0ut);
z4=(14ut,19ut);
z5=(9ut,3ut);
z6=(11ut,0ut);
z7=(19ut,6ut);
z8=(17ut,2ut);
draw z1--z2{z2-z1}..z3{dir 20}..tension1.3..z4{z4-z5}
  & z4--z5{z5-z4}..z6{dir 10}..{dir 55}z7;
labels(range 1 thru 8);
endchar;
beginchar(oct"374",19 ut#,30 ut#,0);       
pickup va_pen;
z1=(0ut,19ut);
z2=(-5ut,3ut);
z3=(0ut,0ut);
z4=(14ut,19ut);
z5=(9ut,3ut);
z6=(11ut,0ut);
z7=(19ut,6ut);
z8=(6ut,24ut);
z9=(11ut,24ut);
z10=(7.5ut,30ut);
z11=(12.5ut,30ut);
draw z1--z2{z2-z1}..z3{dir 20}..tension1.3..z4{z4-z5}
  & z4--z5{z5-z4}..z6{dir 10}..{dir 55}z7;
draw z8--z10;
draw z9--z11;
labels(range 1 thru 12);
endchar;
beginchar("v",17 ut#,19 ut#,0);
pickup va_pen;
z1=(7ut,15ut);
z2=(3ut,5ut);
z3=(9ut,0ut);
z4=(17ut,19ut);
z5=(30ut,15ut);
z6=(0ut,15ut);
z7=(4ut,19ut);
draw z6{dir 50}..z7{right}..{z2-z1}z1..z2{z2-z1}..z3{dir 20}..z4;
labels(range 1 thru 7);
endchar;
beginchar("w",29 ut#,19 ut#,0);
pickup va_pen;
z1=(8ut,15ut);
z2=(5ut,6ut);
z3=(7ut,0ut);
z4=(16ut,7ut);
z5=(19ut,19ut);
z6=(21ut,0ut);
z7=(29ut,19ut);
z8=(42ut,15ut);
z9=(0ut,15ut);
z10=(4ut,19ut);
draw z9{dir 50}..z10{right}..{z2-z1}z1--z2{z2-z1}..z3{right}..z4--z5
  & z5--z4..z6{right}..z7{dir 100};
labels(range 1 thru 8);
endchar;
beginchar("x",19 ut#,19 ut#,0);
pickup va_pen;
z1=(0ut,15ut);
z2=(5ut,19ut);
z3=(9.5ut,13ut);
z4=(7.5ut,6ut);
z5=(0ut,0ut);
z6=(1ut,6ut);
z7=(13ut,10ut);
z8=(16ut,19ut);
z9=(11ut,0ut);
z10=(19ut,6ut);
draw z1{dir 50}..z2{right}..{z4-z3}z3--z4{z4-z3}..z5{left}..{z7-z6}z6;
draw z6--z7{z7-z6}..z8{left}..{z4-z3}z3;
draw z4{z4-z3}..z9{dir 10}..z10{dir 55};
labels(range 1 thru 10);
endchar;
beginchar("y",19 ut#,19 ut#,15ut#);
pickup va_pen;
z1=(0ut,19ut);
z2=(-4.4ut,5ut);
z3=(-2ut,0ut);
z4=(12ut,19ut);
z5=(2ut,-13ut);
z6=(-1ut,-15ut);
z7=(-4ut,-12ut);
z8=(1ut,-5ut);
z9=(19ut,6ut);
draw z1--z2{z2-z1}..z3{dir 20}..z4{z4-z5}
  & z4--z5..z6{left}..z7{up}..z8{z9-z8}..z9{dir 55};
labels(range 1 thru 9);
endchar;
beginchar("z",18 ut#,19 ut#,0ut#);
pickup va_pen;
z1=(0ut,15ut);
z2=(3ut,19ut);
z3=(11ut,17ut);
z4=(16ut,19ut);
z5=(-3ut,0ut);
z6=(0ut,3ut);
z7=(10ut,0ut);
z8=(18ut,6ut);
draw z1{dir 50}..z2{right}..z3..z4{dir 55}
  & z4--z5
  & z5{dir 30}..z6{right}..z7{dir 10}..z8{dir 55};
labels(range 1 thru 8);
endchar;
beginchar("A",35ut#,34ut#,0);
pickup va_pen;
z1=(0ut,2ut);
z2=(12ut,6ut);
z3=(35ut,34ut);
z4=(26ut,0ut);
z5=(12ut,14ut);
z6=(27ut,6ut);
z7=(26ut,28ut);
z8=(22ut,16ut);
z9=(22ut,3ut);
z10=(35ut,6ut);
z11=(3ut,0ut);
draw z1{down}..z11..z2..z7..z3
  & z3--z4;
  draw z6{dir 100}..z8..z5..z9..{dir 55}z10;
labels(range 1 thru 11);
endchar;
beginchar(oct"304",35ut#,44ut#,0);   
pickup va_pen;
z1=(0ut,2ut);
z2=(12ut,6ut);
z3=(35ut,34ut);
z4=(26ut,0ut);
z5=(12ut,14ut);
z6=(27ut,6ut);
z7=(26ut,28ut);
z8=(22ut,16ut);
z9=(22ut,3ut);
z10=(35ut,6ut);
z11=(3ut,0ut);
z12=(32ut,38ut);
z13=(33.5ut,44ut);
z14=(37ut,38ut);
z15=(38.5ut,44ut);
draw z1{down}..z11..z2..z7..z3
  & z3--z4;
  draw z6{dir 100}..z8..z5..z9..{dir 55}z10;
  draw z12..z13;
  draw z14..z15;
labels(range 1 thru 15);
endchar;
beginchar("B",45ut#,34ut#,0);
pickup va_pen;
z1=(8ut,28ut);
z2=(24ut,32ut);
z3=(31ut,28ut);
z4=(31ut,8ut);
z5=(0ut,2ut);
z6=(21ut,34ut);
z7=(23ut,0ut);
z8=(16ut,2ut);
z9=(12ut,6ut);
z10=(6ut,0ut);
z11=(21ut,18ut);
z12=(20ut,30ut);
draw z1..z6{right}..z3{down}..z11{left}
  & z11{right}..z4{down}..z7{left}..z8;
  draw z12--z9..z10..z5;
labels(range 1 thru 12);
endchar;
beginchar("C",19ut#,34ut#,0ut#);
pickup va_pen;
z1=(24ut,26ut);
z2=(20ut,34ut);
z3=(1ut,10ut);
z4=(7ut,0ut);
z5=(19ut,6ut);
z6=(14ut,2ut);
z7=(12ut,20ut);
z8=(0ut,24ut);
z9=(6ut,23ut);
draw z8..z7..z1..z2{left}..z9..z3{z3-z9}..z4{right}..z6..{dir 55}z5;
labels(range 1 thru 9);
endchar;
beginchar("D",39ut#,34ut#,0ut#);
pickup va_pen;
z1=(3ut,30ut);
z2=(21ut,32ut);
z3=(26ut,20ut);
z4=(18ut,0ut);
z5=(15ut,30ut);
z6=(8ut,6ut);
z7=(3ut,0ut);
z8=(0ut,3ut);
z9=(14ut,34ut);
draw z5--z6{z7-z6}..z7..z8..{right}z6{z4-z6}..z4{dir 30}..z3{up}..z2..z9..z1;
labels(range 1 thru 9);
endchar;
beginchar("E",20ut#,34ut#,0ut#);
pickup va_pen;
z1=(23ut,31ut);
z2=(6ut,25ut);
z3=(15ut,18ut);
z4=(0ut,7ut);
z5=(9ut,0ut);
z6=(20ut,6ut);
draw z1..z2{down}..z3
  & z3{left}..z4..z5{right}..{dir 55}z6;   
labels(range 1 thru 6);
endchar;
beginchar("F",32ut#,34ut#,0ut#);
pickup va_pen;
z1=(14ut,34ut);
z2=(12ut,8ut);
z3=(2ut,22ut);
z4=(32ut,34ut);
z5=(9ut,18ut);
z6=(21ut,18ut);
z7=(18ut,31ut);
z8=(6ut,0ut);
z9=(0ut,2ut);
draw z3{up}..z1{right}..z4;
draw z7--z2{z2-z7}..z8..z9;
draw z5--z6;
labels(range 1 thru 9);
endchar;
beginchar("G",25ut#,34ut#,15ut#);
pickup va_pen;
z1=(22ut,26ut);
z2=(19ut,34ut);
z3=(0ut,8ut);
z4=(5ut,0ut);
z5=(20ut,16ut);
z6=(15ut,0ut);
z7=(3ut,-15ut);
z8=(1ut,-11ut);
z9=(8ut,-4ut);
z10=(25ut,6ut);
z11=(0ut,23ut);
draw z11{dir -30}..z1..z2{left}..z3{down}..z4{right}..z5{z5-z6}
  & z5--z6{z6-z5}..z7{left}..z8..z9{z6-z9}..{dir 55}z10;
labels(range 1 thru 11);
endchar;
beginchar("H",30ut#,34ut#,0ut#);
pickup va_pen;
z1=(1ut,29ut);
z2=(5ut,34ut);
z3=(10ut,32ut);
z4=(15ut,34ut);
z5=(10ut,12ut);
z6=(4ut,0ut);
z7=(-1ut,4ut);
z8=(7ut,14ut);
z9=(27ut,24ut);
z10=(31ut,31ut);
z11=(27ut,34ut);
z12=(23ut,31ut);
z13=(17ut,6ut);
z14=(22ut,0ut);
z15=(30ut,6ut);
draw z1..z2{right}..z3..z4;
draw z4--z5{dir 260}..z6..z7..z8{dir 35}--z9{dir 35}..z10..z11..z12--z13;
draw z13{dir 260}..z14{dir 10}..{dir 55}z15;
labels(range 1 thru 15);
endchar;
beginchar("I",29ut#,34ut#,0ut#);
pickup va_pen;
z1=(9ut,34ut);
z2=(25ut,34ut);
z3=(17ut,8ut);
z4=(9ut,0ut);
z5=(1ut,3ut);
z6=(20ut,32ut);
z7=(6ut,31ut);
draw z7{dir 80}..z1{right}..z6..z2{right}
  & z2--z3{z3-z2}..z4{left}..z5;
labels(range 1 thru 7);
endchar;
beginchar("J",22ut#,34ut#,-15ut#);
pickup va_pen;
z1=(9ut,34ut);
z2=(25ut,34ut);
z3=(13ut,0ut);
z4=(4ut,-15ut);
z5=(0ut,-11ut);
z6=(7ut,-4ut);
z7=(22ut,6ut);
z8=(6ut,31ut);
z9=(20ut,32ut);
draw z8..z1{right}..z9..z2{right}
 & z2--z3{z3-z2}...z4{left}..z5{dir 70}..z6{z3-z6}..{dir 55}z7;
labels(range 1 thru 9);
endchar;
beginchar("K",34ut#,34ut#,0ut#);
pickup va_pen;
z1=(5ut,30ut);
z2=(9ut,34ut);
z3=(14ut,32ut);
z4=(19ut,34ut);
z5=(12ut,8ut);
z6=(6ut,0ut);
z7=(0ut,2ut);
z8=(16ut,19ut);
z9=(28ut,30ut);
z10=(35ut,34ut);
z11=(35ut,37ut);
z12=(20ut,6ut);
z13=(26ut,0ut);
z14=(34ut,6ut);
draw z1..z2{right}..z3..z4;
draw z4--z5{z5-z4}..z6..z7;
draw z8{dir 20}..z9{z9-z8}..z10{dir 20};
draw z8--z12{z12-z8}..z13{dir 10}..{dir 55}z14;
labels(range 1 thru 14);
endchar;
beginchar("L",27ut#,34ut#,0ut#);
pickup va_pen;
z1=(6ut,23ut);
z2=(27ut,26ut);
z3=(27ut,6ut);
z4=(20ut,0ut);
z5=(17ut,29ut);
z6=(9ut,7ut);
z7=(4ut,0ut);
z8=(0ut,3ut);
z9=(22ut,34ut);
draw z1{dir 330}..z2..z9..{z6-z5}z5--z6{z6-z5}..z7..z8{up}
     ..z6{right}..z4{dir 10}..{dir 55}z3;
labels(range 1 thru 9);
endchar;
beginchar("M",53ut#,34ut#,0ut#);
pickup va_pen;
z1=(0ut,1ut);
z2=(10ut,6ut);
z3=(31ut,34ut);
z4=(24ut,0ut);
z5=(48ut,34ut);
z6=(40ut,4ut);
z7=(45ut,0ut);
z8=(53ut,6ut);
z9=(49ut,2ut);
z10=(4ut,0ut);
z11=(24ut,26ut);
draw z1..z10..{z11-z2}z2..z11{z11-z2}..z3
 & z3--z4
 & z4--z5
 &z5--z6{z6-z5}..z7{dir 10}..{dir 55}z8;
labels(range 1 thru 11 );
endchar;
beginchar("N",45ut#,34ut#,0ut#);
pickup va_pen;
z1=(0ut,1ut);
z2=(9ut,5ut);
z3=(28ut,34ut);
z4=(23ut,0ut);
z5=(43ut,31ut);
z6=(48ut,34ut);
z7=(4ut,0ut);
z8=(23ut,28ut);
draw z1..z7..{z8-z2}z2..z8{z3-z2}..z3
 & z3--z4
 & z4--z5{z5-z4}..z6{dir 20};
labels(range 1 thru 9);
endchar;
beginchar("O",34ut#,34ut#,0ut#);
pickup va_pen;
z1=(24ut,23ut);
z2=(9ut,0ut);
z3=(1ut,16ut);
z4=(18ut,34ut);
z5=(23ut,34ut);
z6=(28ut,29ut);
z7=(33ut,34ut);
draw z1{down}..z2{left}..z3{dir 80}..z4{right}..cycle;
draw z5.{down}..z6..{dir 50}z7;
labels(range 1 thru 7);
endchar;
beginchar(oct"326",34ut#,46ut#,0ut#);   
pickup va_pen;
z1=(24ut,23ut);
z2=(9ut,0ut);
z3=(1ut,16ut);
z4=(18ut,34ut);
z5=(23ut,34ut);
z6=(28ut,29ut);
z7=(33ut,34ut);
z8=(16.5ut,44ut);
z9=(21.5ut,44ut);
z10=(15ut,38ut);
z11=(20ut,38ut);
draw z1{down}..z2{left}..z3{dir 80}..z4{right}..cycle;
draw z5.{down}..z6..{dir 50}z7;
draw z8..z10;
draw z9..z11;
labels(range 1 thru 11);
endchar;
beginchar("P",32ut#,34ut#,0ut#);
pickup va_pen;
z1=(0ut,2ut);
z2=(6ut,0ut);
z3=(12ut,8ut);
z4=(18ut,30ut);
z5=(4ut,24ut);
z6=(19ut,34ut);
z7=(32ut,26ut);
z8=(18ut,19ut);
draw z1..z2{right}..z3{z4-z3}...z4;
draw z5{dir 60}..z6{right}..{down}z7{down}..{dir 180}z8;
labels(range 1 thru 8);
endchar;
beginchar("Q",37ut#,34ut#,0ut#);
pickup va_pen;
z1=(24ut,23ut);
z2=(9ut,0ut);
z3=(1ut,16ut);
z4=(18ut,34ut);
z5=(10ut,4ut);
z6=(17ut,6ut);
z7=(23ut,0ut);
z8=(37ut,19ut);
z9=(23ut,34ut);
z10=(26ut,29ut);
z11=(32ut,34ut);
draw z1{down}..z2{left}..z3{dir 80}..z4{right}..cycle;
draw z5..z6{right}..z7{dir 40}..{dir 70}z8;
draw z9.{down}..z10..{dir 50}z11;
labels(range 1 thru 11);
endchar;
beginchar("R",36ut#,34ut#,0ut#);
pickup va_pen;
z1=(0ut,2ut);
z2=(6ut,0ut);
z3=(12ut,8ut);
z4=(18ut,30ut);
z5=(4ut,25ut);
z6=(17ut,34ut);
z7=(32ut,26ut);
z8=(18.5ut,19ut);
z9=(21.5ut,18ut);
z10=(26ut,4ut);
z11=(30ut,0ut);
z12=(36ut,6ut);
draw z1..z2{right}..z3{z4-z3}...z4;
draw z5..z6{right}..{dir 270}z7{dir 270}..{dir 185}z8;
draw z9{z12-z9}--z10{z10-z9}..z11{dir 10}..{dir 55}z12;
labels(range 1 thru 12);
endchar;
beginchar("S",27ut#,34ut#,0ut#);
pickup va_pen;
z1=(7ut,24ut);
z2=(22ut,22ut);
z3=(27ut,34ut);
z4=(17ut,27ut);
z5=(15ut,19ut);
z6=(12ut,8ut);
z7=(6ut,0ut);
z8=(0ut,2ut);
draw z1{dir -40}..z2..z3{left}..z4{z6-z4}..z5{z6-z4}..z6{z6-z4}..z7{left}..z8;
labels(range 1 thru 8);
endchar;
beginchar("T",31ut#,34ut#,0ut#);
pickup va_pen;
z1=(15ut,34ut);
z2=(12ut,8ut);
z3=(2ut,24ut);
z4=(31ut,34ut);
z5=(17.5ut,30ut);
z6=(6ut,0ut);
z7=(0ut,2ut);
draw z3{up}..z1{right}..z4;
draw z5--z2{z2-z5}..z6..z7;
labels(range 1 thru 7);
endchar;
beginchar("U",33ut#,34ut#,0ut#);
pickup va_pen;
z1=(7ut,34ut);
z2=(4ut,8ut);
z3=(9ut,0ut);
z4=(24.5ut,15ut);
z5=(29ut,34ut);
z6=(22ut,5ut);
z7=(25ut,0ut);
z8=(0ut,30ut);
z9=(33ut,6ut);
z10=(10ut,31ut);
draw z8{dir 70}..z1{right}..z10{z2-z10}--z2{z2-z10}..z3{dir 15}..z4;
draw z5--z6{z6-z5}..z7{dir 10}..{dir 55}z9;
labels(range 1 thru 10);
endchar;
beginchar(oct"334",33ut#,44ut#,0ut#);   
pickup va_pen;
z1=(7ut,34ut);
z2=(4ut,8ut);
z3=(9ut,0ut);
z4=(24.5ut,15ut);
z5=(29ut,34ut);
z6=(22ut,5ut);
z7=(25ut,0ut);
z8=(0ut,30ut);
z9=(33ut,6ut);
z10=(10ut,31ut);
z11=(19.5ut,44ut);
z12=(18ut,38ut);
z13=(24.5ut,44ut);
z14=(23ut,38ut);
draw z8{dir 70}..z1{right}..z10{z2-z10}--z2{z2-z10}..z3{dir 15}..z4;
draw z5--z6{z6-z5}..z7{dir 10}..{dir 55}z9;
draw z11..z12;
draw z13..z14;
labels(range 1 thru 14);
endchar;
beginchar("V",32ut#,34ut#,0ut#);
pickup va_pen;
z1=(0ut,30ut);
z2=(7ut,34ut);
z3=(10ut,31ut);
z4=(3ut,8ut);
z5=(9ut,0ut);
z6=(20ut,10ut);
z7=(24ut,34ut);
z8=(27ut,29ut);
z9=(32ut,34ut);
draw z1{dir 70}..z2{right}..z3{z4-z3}..z4{z4-z3}..z5{dir 10}..z6{z8-z6}..z7
 & z7{down}..z8..{dir 50}z9;
labels(range 1 thru 9);
endchar;
beginchar("W",55ut#,34ut#,0ut#);
pickup va_pen;
z1=(1ut,30ut);
z2=(8ut,34ut);
z3=(11ut,31ut);
z4=(5ut,8ut);
z5=(10ut,0ut);
z6=(23ut,15ut);
z7=(28ut,34ut);
z8=(23ut,15ut);
z9=(31ut,0ut);
z10=(41ut,12ut);
z11=(41ut,34ut);
z12=(45ut,29ut);
z13=(50ut,34ut);
draw z1{dir 70}..z2{right}..z3{z4-z3}..z4{z4-z3}..z5{right}..z6{z7-z6}..z7;
draw z8{z8-z7}..z9{right}..z10{z12-z10}..z11
 & z11{down}..z12..{dir 50}z13;
labels(range 1 thru 13);
endchar;
beginchar("X",26ut#,34ut#,0ut#);
pickup va_pen;
z1=(6ut,30ut);
z2=(13ut,34ut);
z3=(18ut,22ut);
z4=(15ut,12ut);
z5=(7ut,0ut);
z6=(3ut,6ut);
z7=(16ut,16ut);
z8=(30ut,28ut);
z9=(26ut,34ut);
z10=(20ut,0ut);
z11=(26ut,6ut);
draw z1{dir 60}..z2{right}..z3{z4-z3}..z4{z4-z3}..z5{left}..z6{up}..{dir30}z7;
draw z7{dir 30}..z8{up}..z9{left}..z3{z4-z3}..z4{z4-z3}
     ..z10{dir 10}..{dir 55}z11;
labels(range 1 thru 11);
endchar;
beginchar("Y",29ut#,34ut#,15ut#);
pickup va_pen;
z1=(9ut,31ut);
z2=(1ut,8ut);
z3=(5ut,0ut);
z4=(26ut,34ut);
z5=(16ut,-5ut);
z6=(9ut,-15ut);
z7=(6ut,-12ut);
z8=(9ut,-6ut);
z9=(29ut,6ut);
z10=(0ut,30ut);
z11=(7ut,34ut);
draw z10{up}..z11{right}..z1--z2{z2-z1}..z3{right}...z4{z4-z5}
 & z4--z5{z5-z4}..z6{left}..z7..z8{z9-z8}..{dir 55}z9;
labels(range 1 thru 11);
endchar;
beginchar("Z",29ut#,34ut#,0ut#);
pickup va_pen;
z1=(3ut,29ut);
z2=(11ut,34ut);
z3=(29ut,6ut);
z4=(21ut,0ut);
z5=(23ut,31ut);
z6=(9ut,6ut);
z7=(2ut,0ut);
z8=(0ut,3ut);
z9=(29ut,34ut);
z10=(12ut,18ut);
z11=(26ut,18ut);
draw z1{dir 60}..z2{right}..z5{z9-z5}..z9
  &z9{z6-z9}--z6{z6-z9}..z7..z8..z6{z4-z6}..z4{dir 10}..{dir 55}z3;
draw z10--z11;
labels(range 1 thru 11);
endchar;
beginchar(oct"377",16ut#,34ut#,15ut#);
pickup va_pen;
z1=(-1ut*.87,19ut*.87);
z2=(3ut*.87,29ut*.87);
z3=(-7ut*.87,-18ut*.87);
z4=(13ut*.87,39ut*.87);
z5=(19ut*.87,32ut*.87);
z6=(7ut*.87,21ut*.87);
z7=(18ut*.87,8ut*.87);
z8=(10ut*.87,0ut*.87);
z9=(4ut*.87,2ut*.87);
draw z3--z2{z2-z3}..z4{right}..z5{down}..z6{left}
 & z6{right}..z7{down}..z8{left}..z9;
labels(range 1 thru 9);
endchar;
font_quad 33pt#;
font_normal_space 7.6pt#;    
font_normal_stretch 5pt#;    
font_normal_shrink 3pt#;     
font_x_height 5.7pt#;
vor_m#:=-3ut#;
vor_i#:=-1ut#;
B_e#:=-8ut#;
T_e#:=-2ut#;
boundarychar:=oct"040";
ligtable oct"025": "-"   =:  oct"026";   
ligtable  "a":oct"344":"c":"d":"e":oct"037":"f":"g":"h":"i":"j":
          "k":"l":"m":"n":"p":"t":"u":oct"374":"x":"y":"z":
              "a" |=:| 133,
          oct"344"|=:| 133,
              "b" |=:| 134,
              "c" |=:| 133,
              "d" |=:| 133,
              "f" |=:| 134,
              "g" |=:| 133,
              "h" |=:| 134,
              "i" |=:| 146,
              "j" |=:| 146,
              "k" |=:| 134,
              "l" |=:| 134,
              "m" |=:| 132,
              "n" |=:| 132,
              "o" |=:| 133,
         oct"366" |=:| 133,
              "p" |=:| 146,
              "q" |=:| 133,
              "r" |=:| 132,
              "s" |=:| 146,
         oct"377" |=:| 134,   
              "t" |=:| 134,
              "u" |=:| 146,
         oct"374" |=:| 146,
              "v" |=:| 132,
              "w" |=:| 132,
              "x" |=:| 132,
              "y" |=:| 146,
              "z" |=:| 132;
ligtable "b":"o":oct"366":"r":"v":"w":             
              "e" |=: oct"037",                    
              " " |=:| oct"043",
              "-" |=:| oct"043",
              ")" |=:| oct"043",
              "." |=:| oct"043",
              ":" |=:| oct"043",
         oct"001" |=:| oct"043",    
         oct"047" |=:| oct"043",    
         oct"140" |=:| oct"043",    
         oct"020" |=:| oct"043",    
         oct"021" |=:| oct"043",    
         oct"042" |=:| oct"043",    
              "," |=:| oct"043",
              ";" |=:| oct"043",
              "!" |=:| oct"043",
              "?" |=:| oct"043",
              "a" |=:| oct"044",
          oct"344"|=:| oct"044",
              "b" |=:| oct"045",
              "c" |=:| oct"044",
              "d" |=:| oct"044",
              "f" |=:| oct"045",
              "g" |=:| oct"044",
              "h" |=:| oct"045",
              "i" |=:| oct"043",
              "j" |=:| oct"043",
              "k" |=:| oct"045",
              "l" |=:| oct"045",
              "m" |=:| oct"046",
              "n" |=:| oct"046",
              "o" |=:| oct"044",
         oct"366" |=:| oct"044",
              "p" |=:| oct"043",
              "q" |=:| oct"044",
              "r" |=:| oct"046",
              "s" |=:| oct"043",
         oct"377" |=:| oct"043",
              "t" |=:| oct"045",
              "u" |=:| oct"043",
         oct"374" |=:| oct"043",
              "v" |=:| oct"046",
              "w" |=:| oct"046",
              "x" |=:| oct"046",
              "y" |=:| oct"043",
              "z" |=:| oct"046";
ligtable ||:"(":oct"020":oct"022":"-":     
              "-"  =: oct"025",
              "a" |=:| 130,
         oct"344" |=:| 130,
              "b" |=:| 127,
              "c" |=:| 130,
              "d" |=:| 130,
              "e" |=:| 128,
              "f" |=:| 127,
              "g" |=:| 130,
              "h" |=:| 127,
              "i" |=:| 129,
              "j" |=:| 129,
              "k" |=:| 127,
              "l" |=:| 127,
              "o" |=:| 130,
         oct"366" |=:| 130,
              "p" |=:| 129,
              "q" |=:| 130,
              "s" |=:| 129,
              "t" |=:| 127,
              "u" |=:| 129,
         oct"374" |=:| 129,
              "y" |=:| 129,
         oct"377" |=:| 127;
ligtable "s":oct"377":
              "a" |=:| 149,
          oct"344"|=:| 149,
              "b" |=:| 150,
              "c" |=:| 149,
              "d" |=:| 149,
              "e" |=:| 148,
              "f" |=:| 150,
              "g" |=:| 149,
              "h" |=:| 150,
              "i" |=:| 151,
              "j" |=:| 151,
              "k" |=:| 150,
              "l" |=:| 150,
              "m" |=:| 152,
              "n" |=:| 152,
              "o" |=:| 149,
         oct"366" |=:| 149,
              "p" |=:| 151,
              "q" |=:| 149,
              "r" |=:| 152,
              "s" |=:| 151,
              "t" |=:| 150,
              "u" |=:| 151,
         oct"374" |=:| 151,
              "v" |=:| 152,
              "w" |=:| 152,
              "x" |=:| 152,
              "y" |=:| 151,
              "z" |=:| 152;
ligtable  "A":oct"304":"C":"E":"G":"H":"J":"K":"L":"M":
          "R":"U":"X":"Y":"Z":oct"334":
              "a" |=:| 131,
         oct"344" |=:| 131,
              "b" |=:| 143,
              "c" |=:| 131,
              "d" |=:| 131,
              "f" |=:| 143,
              "g" |=:| 131,
              "h" |=:| 143,
              "i" |=:| 144,
              "j" |=:| 144,
              "k" |=:| 143,
              "l" |=:| 143,
              "m" |=:| 145,
              "n" |=:| 145,
              "o" |=:| 131,
         oct"366" |=:| 131,
              "p" |=:| 144,
              "q" |=:| 131,
              "s" |=:| 144,
              "r" |=:| 145,
              "t" |=:| 143,
              "u" |=:| 144,
         oct"374" |=:| 144,
              "v" |=:| 145,
              "w" |=:| 145,
              "x" |=:| 145,
              "y" |=:| 144,
              "z" |=:| 145;
ligtable  "B":"D":"O":oct"326":"V":"W":"N":
              "a" |=:| 139,
         oct"344" |=:| 139,
              "b" |=:| 140,
              "c" |=:| 139,
              "d" |=:| 139,
              "e" |=:| 153,
              "f" |=:| 140,
              "g" |=:| 139,
              "h" |=:| 140,
              "i" |=:| 141,
              "j" |=:| 141,
              "k" |=:| 140,
              "l" |=:| 140,
              "m" |=:| 142,
              "n" |=:| 142,
              "o" |=:| 139,
         oct"366" |=:| 139,
              "p" |=:| 142,
              "q" |=:| 139,
              "r" |=:| 142,
              "s" |=:| 141,
              "t" |=:| 140,
              "u" |=:| 141,
         oct"374" |=:| 141,
              "v" |=:| 142,
              "w" |=:| 142,
              "x" |=:| 142,
              "y" |=:| 141,
              "z" |=:| 142,
             153 kern B_e#,
             141 kern vor_i#,
             142 kern vor_m#;
ligtable  "F":"I":"P":"S":"T":
              "a" |=:| 135,
         oct"344" |=:| 135,
              "b" |=:| 138,
              "c" |=:| 135,
              "d" |=:| 135,
              "e" |=:| 147,
              "f" |=:| 138,
              "g" |=:| 135,
              "h" |=:| 138,
              "i" |=:| 136,
              "j" |=:| 136,
              "k" |=:| 138,
              "l" |=:| 138,
              "m" |=:| 137,
              "n" |=:| 137,
              "o" |=:| 135,
         oct"366" |=:| 135,
              "p" |=:| 136,
              "q" |=:| 135,
              "r" |=:| 137,
              "s" |=:| 136,
              "t" |=:| 138,
              "u" |=:| 136,
         oct"374" |=:| 136,
              "v" |=:| 137,
              "w" |=:| 137,
              "x" |=:| 137,
              "y" |=:| 136,
              "z" |=:| 137,
              147 kern T_e#,
              137 kern vor_m#;
bye.
}%
\catcode`#=6
\immediate\closeout\lafont%

\usepackage{stmaryrd}
\usepackage{amsmath}
\usepackage{amsthm}
\usepackage{amssymb}
\usepackage{mathrsfs}
\usepackage{times}
\usepackage{color}

\newtheorem{theorem}{Theorem}[section]
\newtheorem{lemma}{Lemma}
\theoremstyle{definition}
\newtheorem{definition}{Definition}

\def\pr{\textrm{\upshape{pr}}}
\def\precond#1{{\vphantom{#1}}^\bullet #1}
\def\postcond#1{{#1}^\bullet}
\def\Production#1{\stackrel{#1}{\Longrightarrow}}
\def\production#1{\stackrel{#1}{\longrightarrow}}
\def\equivalent{\Leftrightarrow}
\newfont{\fsc}{eusm10 scaled 1100}      
\def\powerset#1{\mbox{\fsc P}(#1)}
\def\powermultiset#1{\mathbb{N}^{#1}}
\def\trail#1{\text{~#1}}
\def\into{\rightarrow}
\def\defitem#1{\emph{#1}}
\newcommand{\plat}[1]{\raisebox{0pt}[0pt][0pt]{#1}}   
\newcommand{\inp}{\mathbin\in}                        
\newtheorem{observation}{Observation}{\bfseries}{\rmfamily}
\def\definitionname{Definition}
\newcommand{\refdf}[1]{\definitionname~\ref{df-#1}}
\def\theoremname{Theorem}
\newcommand{\refthm}[1]{\theoremname~\ref{thm-#1}}
\def\lemmaname{Lemma}
\newcommand{\reflem}[1]{\lemmaname~\ref{lem-#1}}
\def\figurename{Figure}
\newcommand{\reffig}[1]{\figurename~\ref{fig-#1}}
\def\isomorph{\approxeq}

\makeatletter
\let\@lemma\lemma
\def\lemma#1#2{\@lemma\label{lem-#1}\def\predecl{#2}\ifx\predecl\myempty\mbox{}\else#2\fi\\}
\let\@definition\definition
\let\@enddefinition\enddefinition
\def\myempty{}
\def\definition#1#2#3{\@definition\label{df-#1}%
  \list{}{\listparindent 0pt\leftmargin 13pt}\item[]%
  \def\predecl{#2}\ifx\predecl\myempty\mbox{}\else\hglue -13pt plus 4pt minus 4pt#2\fi%
  \ifx\begin#3\else\\\fi%
  \def\setleftmargin##1##2{}%
  \def\tmp{#3}\tmp%
}
\def\enddefinition{\endlist\@enddefinition}
\let\@proposition\proposition
\def\proposition#1#2{\@proposition\label{pr-#1}#2\mbox{}\\}
\let\@theorem\theorem
\def\theorem#1#2{\@theorem\label{thm-#1}#2\mbox{}\\}
\let\@conjecture\conjecture
\def\conjecture#1#2{\@conjecture\label{cj-#1}#2\mbox{}\\}
\let\@observation\observation
\def\observation#1#2{\@observation\label{cj-#1}#2\mbox{}\\}
\makeatother

\DeclareFontFamily{T1}{la}{}
\DeclareFontShape{T1}{la}{m}{n}{<->s*[0.8571428571]la14}{}
\DeclareRobustCommand\la{\fontfamily{la}\fontencoding{T1}\selectfont}
\def\processfont#1{\text{\rm\la #1}}
\def\NN{\processfont{N}}
\def\SS{\processfont{S}\,}
\def\TT{\processfont{T}}
\def\FF{\processfont{F}}
\def\MM{\processfont{M}}
\def\OO{\processfont{O}}
\def\PP{\processfont{P}}
\def\ELL{\processfont{l}}

\makeatletter
\def\goesto{\@transition\rightarrowfill}
\def\Goesto{\@transition\Rightarrowfill}
\def\ngoesto{\@transition\nrightarrowfill}
\def\nGoesto{\@transition\nRightarrowfill}
\def\@transition#1{\@ifnextchar[{\@@transition{#1}}{\@@transition{#1}[]}}
\newbox\@transbox
\newbox\@arrowbox
\def\rightarrowfill{$\m@th\mathord-\mkern-6mu%
  \cleaders\hbox{$\mkern-2mu\mathord-\mkern-2mu$}\hfill
  \mkern-6mu\mathord\rightarrow$}
\def\Rightarrowfill{$\m@th\mathord=\mkern-6mu%
  \cleaders\hbox{$\mkern-2mu\mathord=\mkern-2mu$}\hfill
  \mkern-6mu\mathord\Rightarrow$}
\def\@@transition#1[#2]%
   {\setbox\@transbox\hbox
       {\vrule height 1.5ex depth 1ex width 0ex\hskip0.25em$\scriptstyle#2$\hskip0.25em}
   \ifdim\wd\@transbox<1.5em
      \setbox\@transbox\hbox to 1.5em{\hfil\box\@transbox\hfil}\fi
   \setbox\@arrowbox\hbox to \wd\@transbox{#1}
   \ht\@arrowbox\z@\dp\@arrowbox\z@
   \setbox\@transbox\hbox{$\mathop{\box\@arrowbox}\limits^{\box\@transbox}$}
   \ht\@transbox\z@\dp\@transbox\z@
   \mathrel{\box\@transbox}}
\makeatother

\def\ie{i.e.\ }
\def\Act{\textrm{\upshape Act}}
\def\Loc{\textrm{\upshape Loc}}
\def\MVP{\textrm{\upshape MVP}}
\def\concurrent{\smile}

\usepackage{pstricks}
\usepackage{pst-node}
\usepackage{graphicx}

\newcounter{netimage}
\def\p#1:#2;{\cnode #1{0.3}{n\thenetimage-#2}}
\def\P#1:#2;{\p #1:#2;\pscircle*#1{0.1}}
\def\q#1:#2:#3;{\p #1:#2;\rput#1{\rput[l](0.45,0){\large\it #3}}}
\def\Q#1:#2:#3;{\P #1:#2;\rput#1{\rput[l](0.45,0){\large\it #3}}}
\def\qq#1:#2:#3;{\p #1:#2;\rput#1{\rput[t](0,-0.5){\large\it #3}}}
\def\ql#1:#2:#3;{\p #1:#2;\rput#1{\rput[r](-0.45,0){\large\it #3}}}
\def\qt#1:#2:#3;{\p #1:#2;\rput#1{\rput[b](0,0.45){\large\it #3}}}
\def\Ql#1:#2:#3;{\P #1:#2;\rput#1{\rput[r](-0.45,0){\large\it #3}}}
\def\qx#1:#2:#3:#4;{\p #1:#2;\rput#1{\rput#4{\large\it #3}}}
\def\QXX#1:#2:#3:#4:#5;{\p #1:#2;\rput#1{\rput#4{\large\it #3}}\pscircle*#5{0.1}}
\def\s#1:#2:#3;{\p #1:#2;\rput#1{\rput(-0.03,0){\large\it #3}}}
\def\t#1:#2:#3;{\rput#1{\rnode{n\thenetimage-#2}{\psframebox{%
  \vbox to 0.6cm{\vfil\hbox to 0.6cm{\hfil\Large\it #3\hfil}\vfil}}}}}
\def\u#1:#2:#3:#4;{\rput#1{\rnode{n\thenetimage-#2}{\psframebox{%
  \vbox to 0.6cm{\vfil\hbox to 0.6cm{\hfil\Large\it #3\hfil}\vfil}}}}%
  \rput#1{\rput[l](0.6,0){\large\it #4}}}
\def\ut#1:#2:#3:#4;{\rput#1{\rnode{n\thenetimage-#2}{\psframebox{%
  \vbox to 0.6cm{\vfil\hbox to 0.6cm{\hfil\Large\it #3\hfil}\vfil}}}}%
  \rput#1{\rput[b](0,0.6){\large\it #4}}}
\def\ul#1:#2:#3:#4;{\rput#1{\rnode{n\thenetimage-#2}{\psframebox{%
  \vbox to 0.6cm{\vfil\hbox to 0.6cm{\hfil\Large\it #3\hfil}\vfil}}}}%
  \rput#1{\rput[r](-0.6,0){\large\it #4}}}
\def\a#1->#2;{\ncline{->}{n\thenetimage-#1}{n\thenetimage-#2}}
\def\A#1->#2;{\ncarc{->}{n\thenetimage-#1}{n\thenetimage-#2}}
\def\avlinearc{0.2}
\def\av#1[#2]-#3->[#4]#5;{
  \SpecialCoor
  \psline[linearc=\avlinearc]{->}([angle=#2]n\thenetimage-#1)#3([angle=#4]n\thenetimage-#5)
}

\makeatletter
\long\def\petrinet(#1)#2\end{\psscalebox{0.7}{\pspicture(#1)\stepcounter{netimage}#2\endpspicture}\end}

\makeatother

\title{Synchrony vs. Causality\\ in Asynchronous Petri Nets%
\thanks{This work was supported by the DFG (German Research
Foundation), grants {GO-671/6-1} and {NE-1505/2-1}.}}
\author{Jens-Wolfhard Schicke
\institute{Institute for Programming and Reactive Systems, TU Braunschweig, Germany}
\email{drahflow@gmx.de}
\and
Kirstin Peters
\institute{School of EECS, TU Berlin, Germany}
\email{kirstin.peters@tu-berlin.de}
\and
Ursula Goltz
\institute{Institute for Programming and Reactive Systems, TU Braunschweig, Germany}
\email{goltz@ips.cs.tu-bs.de}
}

\def\titlerunning{Synchrony vs Causality in Asynchronous Petri Nets}
\def\authorrunning{J.-W. Schicke, K. Peters \& U. Goltz}

\begin{document}

\maketitle

\begin{abstract}
Given a synchronous system, we study the question whether the behaviour of that
system can be exhibited by a (non-trivially) distributed and hence asynchronous
implementation.
In this paper we show, by counterexample, that synchronous systems cannot in
general be implemented in an asynchronous fashion without either introducing an
infinite implementation or changing the causal structure of the system behaviour.
\\{\bf keywords:} asynchrony, distributed systems, causal semantics, Petri nets
\end{abstract}

\section{Introduction}

It would be desirable -- from a programming standpoint -- to design systems in
a synchronous fashion, yet reap the benefits of parallelism by means of an
(ideally automatically generated) asynchronous implementation executed on
multiple processing units in parallel. We consider the question under which
circumstances such an approach is applicable, or equivalently, what
restrictions must be placed on the synchronous design in order that it may be
simulated asynchronously.

We formalise this problem by means of Petri nets (Section \ref{sec-basic}), a
semi-structural requirement (Section \ref{sec-distr}) on Petri nets to enforce
asynchrony in the implementation, and an equivalence relation (Section
\ref{sec-cpt}) on possible Petri net behaviours to decide whether a candidate
implementation is indeed faithful to the synchronous specification.

\begin{figure}[t]
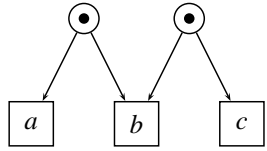

  \begin{center}
    \begin{petrinet}(6,4)
      \P(2,3):p;
      \P(4,3):q;
      \t(1,1):a:a;
      \t(3,1):b:b;
      \t(5,1):c:c;

      \a p->a;
      \a p->b;
      \a q->b;
      \a q->c;
    \end{petrinet}
  \end{center}

  \caption{A fully reached, pure \textbf{M}, the problematic structure from \cite{glabbeek08syncasyncinteraction}}
  \label{fig-pure-M}
\end{figure}

Countless equivalence relations for system behaviour have already been proposed.
When comparing the strictness of these equivalences, as done in
\cite{vanglabbeek93linear} or \cite{vanglabbeek01refinement},
and exploring the resulting lattice,
one finds multiple ``dimensions'' of features along which such an equivalence
may be more or less discriminating.
The most prominent one is the linear-time branching-time axis, denoting how well
the decision structure of a system is captured by the equivalence.
Another dimension relevant to this paper is that along which the detail of
the causal structure increases.
On the first of these two dimensions, we would at the very least like to detect
deadlocks introduced by the implementation, on the second one, at least a
reduction in concurrency due to the implementation.
As every (non-trivial) implementation will introduce internal $\tau$-transitions,
a suitable equivalence must abstract from them, as long as they do not
allow a divergence.

\cite{glabbeek08syncasyncinteraction} answers part of the question of
distributed implementability for a certain equivalence of this spectrum, namely step
readiness equivalence.
Step readiness equivalence is one of the weakest equivalences that
respects branching time, concurrency and divergence
to some degree but
abstracts from internal actions. For this equivalence we derived an
exact characterisation of asynchronously implementable (``distributable'')
Petri nets.
The main difficulty in implementing arbitrary Petri nets up to step readiness
equivalence is a structure called pure \textbf{M}, depicted in \reffig{pure-M},
where two parallel transitions are in pairwise conflict with a common third.
By \cite{glabbeek08syncasyncinteraction} a synchronous net is distributable
only if it contains no fully reachable pure \textbf{M}. The other direction needed
for exactness has not been published yet, as the only as of yet existing proofs
utilises an infinite implementation.

Using the strictly weaker completed step trace equivalence,
\cite{schicke09synchrony} proved any synchronous net to be distributable.
Comparing these two results and the given implementation in the latter we made
a very interesting observation: We were unable to find an implementation of a
synchronous net with a fully reachable pure \textbf{M} which did not introduce
additional causal dependencies.

In this paper we show that this drawback holds for any sensible encoding of
synchronous interactions, i.e., it is a general phenomenon of encoding
synchrony. We reach that result by extending the pure \textbf{M} of \reffig{pure-M}
into a repeated pure \textbf{M}, depicted in \reffig{counterexample}.
We thereby get a separation result similar to
\cite{glabbeek08syncasyncinteraction} along a different, namely the causal,
dimension of the spectrum of behavioural equivalences.

We introduce basic Petri net concepts in Section \ref{sec-basic}, then turn to
recounting the definition of distributability in Section \ref{sec-distr}. Afterwards we
introduce completed pomset trace equivalence in Section \ref{sec-cpt}, justify it by
means of illustrative examples, and use it in Section \ref{sec-impossible} to prove the
impossibility of implementing general Petri nets while respecting causality.
Finally Section \ref{sec-conclusion} concludes.

\begin{figure}[t]
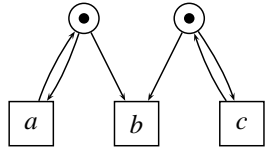

  \begin{center}
    \begin{petrinet}(6,4)
      \P(2,3):p;
      \P(4,3):q;
      \t(1,1):a:a;
      \t(3,1):b:b;
      \t(5,1):c:c;

      \A p->a; \A a->p;
      \a p->b;
      \a q->b;
      \A q->c; \A c->q;
    \end{petrinet}
  \end{center}
  \caption{A repeated pure \textbf{M}. A finite, 1-safe, undistributable net used as a running counterexample.}
  \label{fig-counterexample}
\end{figure}

\clearpage
\section{Basic Notions}\label{sec-basic}

Most material in this section has been taken verbatim or with minimal adaptation
from \cite{glabbeek08syncasyncinteraction} or \cite{schicke09synchrony}.

\noindent
Where dealing with tuples, we use $\pr_1, \pr_2, \ldots$ as the projection
functions returning the first, second, \ldots\linebreak[4] element respectively. We extend these
functions to sets element-wise.

\begin{definition}{net}{
  Let \Act{} be a set of \emph{visible actions} and
  $\tau\mathbin{\not\in}\Act$ be an \emph{invisible action}.
  }
  A \defitem{labelled net} (over \Act) is a tuple
  $N = (S, T, F, M_0, \ell)$ where
  \begin{itemize}
    \item $S$ is a set (of \defitem{places}),
    \item $T$ is a set (of \defitem{transitions}),
    \item $F \subseteq
       S \times T \cup T \times S$
      (the \defitem{flow relation}),
    \item $M_0 \subseteq S$ (the \defitem{initial marking}) and
    \item \plat{$\ell: T \into \Act \cup \{\tau\}$} (the \defitem{labelling function}).
  \end{itemize}

  A net is called \defitem{finite} iff $S$ and $T$ are finite.
\end{definition}

\noindent
Petri nets are depicted by drawing the places as circles, the
transitions as boxes containing the respective label, and the flow
relation as arrows (\defitem{arcs}) between them.
When a Petri net represents a concurrent system, a global state of such a system is
given as a \defitem{marking}, a set of places, the initial state being
$M_0$.  A marking is depicted by placing a dot (\defitem{token}) in
each of its places.  The dynamic behaviour of the represented system
is defined by describing the possible moves between markings. A
marking $M$ may evolve into a marking $M'$ when a nonempty set of transitions
$G$ \defitem{fires}. In that case, for each arc $(s,t) \in F$ leading
to a transition $t$ in $G$, a token moves along that arc from $s$ to
$t$.  Naturally, this can happen only if all these tokens are
available in $M$ in the first place. These tokens are consumed by the
firing, but also new tokens are created, namely one for every outgoing
arc of a transition in $G$. These end up in the places at the end of
those arcs.  A problem occurs when as a result of firing $G$ multiple
tokens end up in the same place. In that case $M'$ would not be
a marking as defined above. In this paper we restrict
attention to nets in which this never happens. Such nets are called
\defitem{1-safe}.  Unfortunately, in order to formally define this
class of nets, we first need to correctly define the firing rule
without assuming 1-safety. Below we do this by forbidding the firing
of sets of transitions when this might put multiple tokens in the same
place.

To help track causality throughout the evolution of a net, we extend the usual
notion of marking to \defitem{dependency marking}. Within these dependency
markings, every token is augmented with the labels of all transitions having
causally contributed to its existence.
The other basic Petri net notions
presented here have been extended in the same manner.
While it might seem more natural
to annotate the causal history of the tokens by a partial order, we only
use a set here in order to keep the number of reachable markings finite
for finite nets (a property a later proof will utilise).

We denote the preset and postset of a net element $x\in S\cup T$ by
$\precond{x} := \{y \mid (y, x) \in F\}$
and
$\postcond{x} := \{y \mid (x, y) \in F\}$
respectively.
These functions are extended to sets in the usual manner, \ie
$\precond{X} := \{y \mid y \inp \precond{x},~ x \inp X\}$.

\begin{definition}{steps}{
  Let $N = (S, T, F, M_0, \ell)$ be a net.
  Let $M_1, M_2 \subseteq S\times\powerset{\Act}$.
  }
  $G \subseteq T,
  G \not= \varnothing$, is called a \defitem{dependency step from $M_1$ to $M_2$},
  $M_1 [G\rangle_N M_2$, iff
  \begin{itemize}
    \item all transitions contained in $G$ are enabled, i.e.
      \begin{equation*}
        \forall t\in G. \precond{t} \subseteq \pr_1(M_1) \wedge 
          (\pr_1(M_1) \setminus \precond{t}) \cap \postcond{t} =
          \varnothing \trail{,}
      \end{equation*}
    \item all transitions of $G$ are independent, that is not conflicting:
      \begin{align*}
        &\forall t,u \in G, t\not= u. \precond{t} \cap \precond{u} = \varnothing
        \wedge \postcond{t} \cap \postcond{u} = \varnothing \trail{,}
      \end{align*}
    \item causalities are extended by the labels of the firing transitions:
      \begin{align*}
        M_2 ={}& \left\{p \in M_1 \mid \pr_1(p) \not\in \precond{G}\right\} \cup \\
        &\left\{\left(s, (\{\ell(t)\} \setminus \{\tau\}) \cup \bigcup_{\hspace{-4pt}p \in M_1 \wedge \pr_1(p) \in \precond{t}\hspace{-2.5em}}
            \pr_2(p)\right) \middle|\,
            t \in G, s \in \postcond{t}\right\}
        \trail{.}
      \end{align*}
  \end{itemize}
\end{definition}

\noindent
Applying $\pr_1$ to a dependency marking results in the classical Petri net
notion of marking and similar for the other notions introduced in this section.
We will however mainly employ the versions defined here and drop the qualifier
``dependency'' most of the time.
A token $(s, P) \in M$ is $Q$-dependent iff $Q \subseteq P$ and $Q$-independent
iff $P \cap Q = \varnothing$.

To simplify the following argumentation we use some abbreviations.
$\production{\mu}_N$ denotes a labelled step on a single transition labelled
$\mu$. $\Goesto[\,a~]_N$ denotes a step on $a$ surrounded by arbitrary
$\tau$-steps, i.e., $\Goesto[]_N$ abstracts from $\tau$-steps.

\begin{definition}{steprel}{
  Let $N = (S, T, F, M_0, \ell)$ be a labelled net.
  }
  \vspace{-3ex}
    \item  We extend the labelling function $\ell$ to (multi)sets element-wise.
    \item $\mathord{\production{}_N} \subseteq \powerset{S\times\powerset{\Act}} \times
      \powermultiset{\Act} \times \powerset{S\times\powerset{\Act}}$
      is given by\\ $~\qquad M_1 \production{A}_N M_2 \equivalent
      \exists\, G \subseteq T. M_1~[G\rangle_N~ M_2 \wedge A = \ell(G)$
    \item $\mathord{\production{\tau}_N} \subseteq \powerset{S\times\powerset{\Act}} \times
      \powerset{S\times\powerset{\Act}}$
      is defined by\\ $~\qquad M_1 \production{\tau}_N M_2 \equivalent
      \exists t \inp T. \ell(t) \mathbin= \tau \wedge M_1 ~[\{t\}\rangle_N~ M_2$
    \item $\mathord{\Production{}_N} \subseteq \powerset{S\times\powerset{\Act}} \times \Act^* \times
      \powerset{S\times\powerset{\Act}}$ is defined by\\
      $~\qquad M_1 \Goesto[\,a_1 a_2 \cdots a_n~]_N M_2 \equivalent
      M_1
      \production{\tau}^*_N \production{\{a_1\}}_N
      \production{\tau}^*_N \production{\{a_2\}}_N
      \production{\tau}^*_N \cdots
      \production{\tau}^*_N \production{\{a_n\}}_N
      \production{\tau}^*_N
      M_2$\\[3pt]
      where $\production{\tau}^*_N$ denotes the reflexive and
      transitive closure of $\production{\tau}_N$.

  We omit the subscript $N$ if clear from context.
\end{definition}
\noindent
We write $M_1 \production{A}_N$ for $\exists M_2. M_1
\production{A}_N M_2$, $M_1 \arrownot\production{A}_N$ for $\nexists M_2. M_1
\production{A}_N M_2$ and similar for the other two relations.
Likewise $M_1 [G\rangle_N$ abbreviates $\exists M_2. M_1 [G\rangle_N M_2$.
A marking $M_1$ is said to be \defitem{reachable} iff there is a
sequence of labels $\sigma \in \Act^*$ such that $M_0 \times \{\emptyset\} \Production{\sigma}_N M_1$. The set of all
reachable markings is denoted by $[M_0\rangle_N$.

As said before, here we only want to consider 1-safe nets. Formally,
we restrict ourselves to \defitem{contact-free nets}, where in every
reachable marking $M_1 \in [M_0\rangle$ for all $t \in T$ with
  $\precond{t} \subseteq \pr_1(M_1)$\vspace{-1ex}
\begin{equation*}
  (\pr_1(M_1) \setminus \precond{t}) \cap \postcond{t} = \varnothing \trail{.}
\vspace{1pt}
\end{equation*}
For such nets, in \refdf{steps} we can just as well consider a
transition $t$ to be enabled in $M$ iff $\precond{t}\subseteq \pr_1(M)$, and
two transitions to be independent when $\precond{t} \cap \precond{u} =
\varnothing$.

\section{Distributed Nets}\label{sec-distr}

After having introduced Petri nets in general, we still need to find a
notion of such a net being distributed before being able to answer the
question of distributed implementability. A straightforward approach is
to assign to each net element a \defitem{location}, place sensible restrictions
on arrows crossing location borders, and restrict the sets of net elements being
allowed to reside on the same location.

We will regard locations as sequential execution units of the
underlying system, each one able to execute at most one action during each step. This
necessitates that no pair of transitions firing in the same step can reside on
the same location. Additionally, if locations are indeed physically apart as
their name suggests, communication between them can only proceed asynchronously.

We discussed a very similar notion of distribution in
\cite{glabbeek08syncasyncinteraction}, whence
the following description and definition of the present version have been derived from.
The central insight from that paper is that the synchronous removal of tokens
from preplaces of a transition is essential to the conflict resolution taking
place between multiple enabled transitions and that hence transitions must
reside on the same location as their preplaces.

We model the association of locations to the places and transitions
in a net $N=(S,T,F,M_0,\ell)$ as a function $D: S\cup T \rightarrow
\Loc$, with $\Loc$ a set of possible locations.  We refer to such a
function as a \defitem{distribution} of $N$.  Since the identity of
the locations is irrelevant for our purposes, we can just as well
abstract from $\Loc$ and represent $D$ by the equivalence relation
$\equiv_D$ on $S\cup T$ given by $x \equiv_D y$ iff $D(x)=D(y)$.

\begin{definition}{distributed}{
  Let $N = (S, T, F, M_0, \ell)$ be a net.
  }
  The \defitem{concurrency relation} $\mathord{\concurrent} \!\subseteq\!
  T^2$ is given by $t \!\concurrent\! u \equivalent t \!\ne\! u \wedge \exists
  M \inp [M_0\rangle.\! M [\{t,\!u\}\rangle$.
  $N$ is \defitem{distributed} iff it has a distribution $D$ such that
  \begin{itemize}
    \item
      $\forall s \in S, ~t \in T.\hspace{1pt}s \in \precond{t} \implies t \equiv_D s$,
    \item
     $t \concurrent u \implies t\not\equiv_D u$.
  \end{itemize}
\end{definition}

\noindent
It is straightforward to give a semi-structural\footnote{mainly structural, but with a reachability side-condition}
characterisation of this class of nets:

\begin{observation}{distributed}{}
A net is distributed iff there is no sequence $t_0,\ldots,t_n$ of
transitions with $t_0 \smile t_n$ and
$\precond{t_{i-1}}\cap\precond{t_{i}}\neq\emptyset$ for $i=1,\ldots,n$.
\end{observation}

\section{Completed Pomset Trace}\label{sec-cpt}

We now motivate the equivalence relation used for the rest of the paper
by means of highlighting some possible shortcomings of implementations one
would intuitively like to avoid.

\begin{figure}[t]
  \begin{center}
    \begin{petrinet}(10,6)
      \P(4,5):px2;
      \P(6,5):qy2;
      \t(3,3.5):ta:$\tau$;
      \t(5,3.25):tb:$\tau$;
      \t(7,3.5):tc:$\tau$;
      \t(1,1):a:a;
      \t(5,1):b:b;
      \t(9,1):c:c;
      \p(2,2.25):pa;
      \p(5,2):pb;
      \p(8,2.25):qc;
      
      \P(5,4.5):l;

      \a pa->a; \a pb->b;
      \a qc->c; \a qb->b;
      \av a[90]-(1,5)->[180]px2;
      \av c[90]-(9,5)->[0]qy2;

      \a px2->ta; \a px2->tb; \a qy2->tc; \a qy2->tb;
      \a ta->pa; \a tb->pb; \a tc->qc;

      \A ta->l; \A l->ta;
      \A tb->l; \A l->tb;
      \A tc->l; \A l->tc;

      \psline[linestyle=dotted](2.2,6)(2.2,3)(2.5,2.7)(7.5,2.7)(7.8,3)(7.8,6)
      \psline[linestyle=dotted](4,2.6)(4,0)
      \psline[linestyle=dotted](6,2.6)(6,0)
    \end{petrinet}
  \end{center}
  \caption{A centralised implementation of \reffig{counterexample}, location borders dotted.}
  \label{fig-central}
\end{figure}

When trying to implement a synchronous Petri net by a distributed one, one
of the easiest approaches is central serialisation of the entire original net
by introduction of a single new place connected with loops to every transition,
thereby vacuously fulfilling the requirement that no parallel transitions may
reside on the same location. This clearly loses parallelism. We illustrate in
\reffig{central} the result of applying a slightly more intricate variant of this scheme,
where every
visible step of the original still exists in the implementation,
to the repeated pure \textbf{M}.
Nonetheless,
this approach is intuitively not scalable, as all decisions made concurrently
in the original net are now made in sequence.
In particular, the parts of the net firing $a$ were completely independent of
those parts firing $c$ in the specification, while being connected trough the
central place in the implementation. Such new dependencies can be detected if
the causal dependencies between events are included in the behavioural description
of a net. Apart from the obvious implications for scalability, if a Petri net is
used as an abstract description of a more concrete system, a new dependency
might enable interactions between different parts of the system the designer
did not take into account.
Hence we would like to disallow such a strategy
by means of the equivalence between specification and implementation.

No such causalities are introduced by the implementation in \reffig{deadlock}.
There however, one of the cycles of $a$'s or $c$'s may spontaneously decide to
commit to the $b$ action and wait until the other does likewise, resulting in
what is essentially a local deadlock. Compared to the original net, where
$a$ stayed enabled until $b$ was fired, such behaviour is new.
Trying to resolve this deadlock by adding a $\tau$-transition in the reverse
direction would introduce a diverging computation not present in the original net.

\begin{figure}[t]
  \begin{center}
    \begin{petrinet}(10,4)
      \p(2,3):pa;
      \p(4,3):pb;
      \p(6,3):qb;
      \p(8,3):qc;
      \t(1,1):a:a;
      \t(5,1):b:b;
      \t(9,1):c:c;

      \t(3,3):pb1:$\tau$;
      \t(7,3):qb1:$\tau$;

      \A pa->a; \a pb->b;
      \A qc->c; \a qb->b;
      \A a->pa; \A c->qc;

      \a pa->pb1; \a pb1->pb;
      \a qc->qb1; \a qb1->qb;

      \psline[linestyle=dotted](3.5,4)(3.5,0)
      \psline[linestyle=dotted](6.5,4)(6.5,0)
    \end{petrinet}
  \end{center}
  \caption{A locally deadlocking implementation of \reffig{counterexample}, location borders dotted.}
  \label{fig-deadlock}
\end{figure}

All these deviations from the original behaviour can elegantly be captured
by the causal equivalence from \cite{schicke09synchrony}, called
completed pomset trace equivalence. It extends the
pomset trace equivalence of \cite{pratt85pomset} as to detect local
deadlocks, which can be regarded as unjust executions in the sense of
\cite{reisig84partialorder}.

Pomset trace equivalence is obtained by unrolling a Petri net into a
process as defined by \cite{petri77nonsequential}. Such a process can be understood
to be an account of one particular way to decide all conflicts which occurred
while proceeding from one marking to the next. The behaviour of the net
is hence a set of these processes, covering all possible ways to decide conflicts.

Unrolling a net $N$ intuitively proceeds as follows: The initially marked
places of $N$ are copied into a new net $\NN$ and their correspondence to the original
places recorded in a mapping $\pi$. Then, whenever in $N$ a transition $t$ is fired,
this is replayed in $\NN$ by a new transition
connected to places corresponding by $\pi$ to the original preplaces of $t$
and which are not yet connected to any other post-transition.
A new place of $\NN$ is created
for every token produced by $t$. Again all correspondences are recorded in $\pi$.
Every place of $\NN$ has thus at most one post-transition. If it has none, this place
represents a token currently being placed on the corresponding original place.

As a shorthand notation to gather these places, we introduce the \defitem{end}
of a net.

\begin{definition}{netend}{
  Let $N = (S, T, F, M_0, \ell)$ be a labelled net.
  }
  The \defitem{end} of the net is defined as
  $N^\circ := \{s \in S \mid \postcond{s} = \varnothing\}$.
\end{definition}

\begin{definition}{process}{}
 \vspace{-3ex}
 \item[]
  A pair $\PP = (\NN, \pi)$ is a
  \defitem{process} of a net $N = (S, T, F, M_0, \ell)$
  iff
  \begin{itemize}\itemsep 3pt
   \item $\NN = (\SS, \TT, \FF, \MM_0, \ELL)$ is a net, satisfying
   \begin{itemize}
    \item $\forall s \in \SS. |\precond{s}| \leq\! 1\! \geq |\postcond{s}|
    \wedge\, s \in \MM_0 \equivalent \precond{s}=\emptyset$
    \item $\FF$ is acyclic, \ie
      $\forall x \inp \SS \cup \TT. (x, x) \mathbin{\not\in} \FF^+$,\\
      where $\FF^+$ is the transitive closure of $\{(t,u)\mid \FF(t,u)>0\}$,
    \item and $\{t \mid (t,u)\in \FF^+\}$ is finite for all $u\in \TT$.
   \end{itemize}
    \item $\pi:\SS \cup \TT \rightarrow S \cup T$ is a function with 
    $\pi(\SS) \subseteq S$ and $\pi(\TT) \subseteq T$, satisfying
   \begin{itemize}
    \item $s \in M_0 \equivalent |\pi^{-1}(s) \cap \MM_0| = 1$ for all $s\in S$,
    \item $\pi$ is injective on $\MM_0$,
    \item $\forall t \in \TT, s \in S.
      F(s, \pi(t)) = |\pi^{-1}(s) \cap \precond{t}| \wedge
      F(\pi(t), s) = |\pi^{-1}(s) \cap \postcond{t}|$, and
    \item $\forall t \in \TT. \ELL(t) = \ell(\pi(t))$.\footnote{While $\ell$ and $\ELL$ look nearly identical, the authors see no problem in that, given the close correspondence.}
  \end{itemize}
  \end{itemize}
  $\PP$ is called \defitem{finite} if $\NN$ is finite.

  $\PP$ is \defitem{maximal} iff $\pi(\NN^\circ) \arrownot\production{}_N$.
  The set of all maximal processes of a net $N$ is denoted by $MP(N)$.
\end{definition}

\noindent
To disambiguate between a not-yet-occurred firing of a transition $a$ and
the impossibility of firing an $a$, we restrict the set of processes relevant
for the behavioural description to maximal processes. We thereby obtain
a just semantics in the sense of \cite{reisig84partialorder}, i.e. a transition
which remained enabled infinitely long must ultimately fire.

To abstract from the $\tau$-actions introduced in an implementation, we
extract from the maximal processes the causal structure between the fired
visible events in the form of a partially ordered multiset (\defitem{pomset}).
Formally, a pomset is an isomorphism class of a partially ordered multiset of action labels.

\begin{definition}{lpo}{}
  A \defitem{labelled partial order} is a structure $(V, T, \leq, l)$
  where
  \begin{itemize}
    \item $V$ is a set (of \defitem{vertices}),
    \item $T$ is a set (of \defitem{labels}),
    \item $\leq \,\, \subseteq V \times V$ is a partial order relation and
    \item $l: V \rightarrow T$ (the \defitem{labelling function}).
  \end{itemize}

  Two labelled partial orders $o = (V, T, \leq, l)$ and $o' = (V', T, \leq', l')$ are
  \defitem{isomorphic}, $o \isomorph o'$, iff there exist a bijection
  $\varphi: V \rightarrow V'$ such that
  \begin{itemize}
    \item $\forall v \in V. l(v) = l'(\varphi(v))$ and
    \item $\forall u, v \in V. u \leq v \equivalent \varphi(u) \leq' \varphi(v)$.
  \end{itemize}
\end{definition}

\begin{definition}{pomset}{
  Let $o = (V, T, \leq, l)$ be a partial order.
  }
  The \defitem{pomset} of $o$ is its isomorphism class $[o] := \{o' \mid o \isomorph o'\}$.
\end{definition}

\noindent
By hiding the unobservable transitions of a process, we gain a pomset which
describes causality relations of all participating visible transitions.

\begin{definition}{mvp}{
  Let $\PP = ((\SS, \TT, \FF, \MM_0, \ELL), \pi)$ be a process.
  }
  Let $\OO := \{t \in \TT \mid \ELL(t) \ne \tau\}$, i.e. the visible transitions of the process.
  The \defitem{visible pomset} of $\PP$ is the pomset
  $VP(\PP) := [(\OO, \Act, \FF^* \cap \OO \times \OO, \ELL \cap (\OO \times \Act))]$
  where $\FF^*$ is the
  transitive and reflexive closure of the flow relation $\FF$.
  
  $\MVP(N) := \{VP(\PP) \mid \PP \in MP(N)\}$ is the set of pomsets of all maximal
  processes of $N$.
\end{definition}

\noindent
Using this notion we can now define completed pomset trace equivalence.

\begin{definition}{jpte}{}
  Two nets $N$ and $N'$ are \defitem{completed pomset trace equivalent}, $N \simeq_{CPT} N'$, iff
  $\MVP(N) = \MVP(N')$.
\end{definition}

\section{Impossibility}\label{sec-impossible}

\begin{figure}[t]
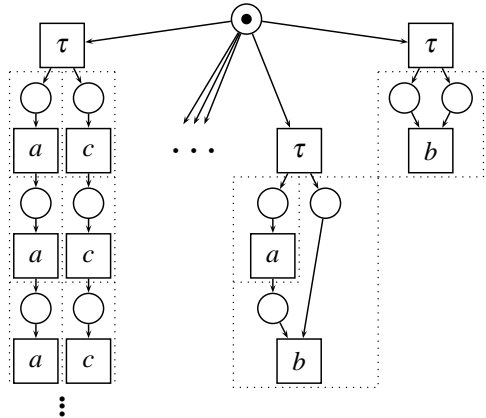

  \begin{center}
    \begin{petrinet}(10,8)
      \P(5,7.5):p1;

      \t(1.5,7):t1:$\tau$;
      \p(1,6):p1-1;
      \p(2,6):p1-2;
      \t(1,5):x1-1:a;
      \t(2,5):y1-1:c;
      \p(1,4):p1-3;
      \p(2,4):p1-4;
      \t(1,3):a1-1:a;
      \t(2,3):c1-1:c;
      \p(1,2):p1-5;
      \p(2,2):p1-6;
      \t(1,1):x1-2:a;
      \t(2,1):y1-2:c;
      \rput(1.5,0.25){\Huge $\vdots$}
      \a p1->t1; \a t1->p1-1; \a t1->p1-2;
      \a p1-1->x1-1; \a p1-2->y1-1;
      \a x1-1->p1-3; \a y1-1->p1-4;
      \a p1-3->a1-1; \a p1-4->c1-1;
      \a a1-1->p1-5; \a c1-1->p1-6;
      \a p1-5->x1-2; \a p1-6->y1-2;

      \rput(4,5){\Huge $\hdots$}
      \pnode(4,5.5){n\thenetimage-dots1}
      \pnode(3.8,5.5){n\thenetimage-dots1l}
      \pnode(4.2,5.5){n\thenetimage-dots1r}
      \a p1->dots1;
      \a p1->dots1l;
      \a p1->dots1r;

      \t(6,5):t3:$\tau$;
      \p(5.5,4):p3-1;
      \p(6.5,4):p3-2;
      \t(5.5,3):a3-1:a;
      \p(5.5,2):p3-3;
      \t(6,1):b3-1:b;

      \a p1->t3;
      \a t3->p3-1; \a t3->p3-2;
      \a p3-1->a3-1;
      \a a3-1->p3-3;
      \a p3-3->b3-1;
      \a p3-2->b3-1;

      \t(8.5,7):t2:$\tau$;
      \p(8,6):p2-1;
      \p(9,6):p2-2;
      \t(8.5,5):b2-1:b;
      \a p1->t2; \a t2->p2-1; \a t2->p2-2;
      \a p2-1->b2-1; \a p2-2->b2-1;

      \psline[linestyle=dotted](0.5,6.5)(2.5,6.5)
      \psline[linestyle=dotted](7.5,6.5)(9.5,6.5)
      \psline[linestyle=dotted](4.75,4.5)(7.5,4.5)
      \psline[linestyle=dotted](7.5,6.5)(7.5,0.5)(4.75,0.5)(4.75,4.5)
      \psline[linestyle=dotted](7.5,4.5)(9.5,4.5)(9.5,6.5)
      \psline[linestyle=dotted](0.5,4.5)(2.5,4.5)
      \psline[linestyle=dotted](0.5,2.5)(2.5,2.5)
      \psline[linestyle=dotted](0.5,6.5)(0.5,0.5)
      \psline[linestyle=dotted](1.5,6.5)(1.5,0.5)
      \psline[linestyle=dotted](2.5,6.5)(2.5,0.5)
      \psline[linestyle=dotted](4.75,2.5)(6.0,2.5)(6.0,4.5)
    \end{petrinet}
  \end{center}
  \caption{An infinite implementation of \reffig{counterexample}, constructed by taking every maximal process and
    initially choosing one, location borders dotted.}
  \label{fig-infinite}
\end{figure}

As completed pomset trace equivalence is a very linear-time equivalence,
it disregards the decision structure of a system and
an implementation like the one of \reffig{infinite}, which simply provides a separate branch for each possible maximal process of the original net, would be fully
satisfactory. In practice though, such an infinite implementation is unwieldy to
say the least. If however infinite implementations are ruled out,
our main result shows that
no valid implementation
of the repeated pure \textbf{M} of \reffig{counterexample} exists.

Before we consider this main theorem of the paper, let us concentrate on two
auxiliary lemmata. The first
states that the careful introduction of a
$\tau$-transition before an arbitrary transition of a net, as described below,
does not significantly influence the properties of that net.

\begin{lemma}{implementationrefinement}{
  Let $N = (S, T, F, M_0, \ell)$ be a finite, 1-safe, distributed net with
  the distribution function $D$. Let $t \in T$.
  }
  The net $N' = (S', T', F', M_0, \ell')$ with
  \begin{itemize}\vspace{-1ex}
    \item $S' = S \cup \{s_t\}$,
    \item $T' = T \cup \{\tau_t\}$,
    \item $F' = (F \setminus (S \times \precond{t})) \cup
      \{(s, \tau_t) \mid s \in \precond{t}\} \cup
      \{(\tau_t, s_t), (s_t, t)\}$, and
    \item $\ell'(x) = \begin{cases}
        \tau&\text{ if } x = \tau_t\\
        \ell(x)&\text{ otherwise}
      \end{cases}$
  \end{itemize}
  is finite, 1-safe, distributed and completed pomset trace equivalent to $N$.
\end{lemma}

\begin{proof}(Sketch)\\
  $N'$ is finite as only two new elements were introduced.

  $N'$ is completed pomset trace equivalent to $N$. Given a process $(\NN,
  \pi)$ of $N$, a process of $N'$ can be constructed by refining in $\NN$ every
  transition $u$ in the same manner as $\pi(u)$ was in $N$.
  For the reverse direction, note that in every maximal processes of $N'$,
  $\pi(u) = t \implies \pi(\precond{u}) = \{s_t\} \wedge \pi(\precond{s_t}) = \{\tau_t\}$.
  By fusing $u$, $\precond{u}$, and ${}^\bullet{\precond{u}}$ into a single
  transition $v$ whenever $\pi(u) = t$ and setting the process mapping of $v$ to $t$,
  a maximal process of $N'$ can be transformed into a maximal process of $N$.

  For the same reason, $N'$ is also 1-safe.

  \noindent
  $N'$ is distributed with the distribution function
  $D'(x) := \begin{cases}
    D(t)&\text{ if } x = s_t \vee x = \tau_t\\
    D(x)&\text{ otherwise }
  \end{cases}$.
  The places in $\precond{\tau_t}$ are on $D(t) = D'(\tau_t)$. $D'(s_t) = D(t) = D'(t)$.
  Hence all transitions are on the same location as their preplaces.
  No new parallelism is introduced, as a parallel firing of either $\tau_t$ or $t$
  with some other transition $u$ can only occur if $t$ and $u$ could already fire in parallel
  in $N$.
\end{proof}

\noindent
Next we show, that if a marking is reached twice during an execution,
the dependencies of all tokens consumed and produced by a transition firing in
such a cycle are equal.

\begin{lemma}{dependenceinloop}{
  Let $N = (S, T, F, M_0, \ell)$ be a finite, 1-safe net.
  Let $t_s, t_{s+1}, \ldots, t_{e-1}, t_e \in T$ be a sequence of transitions
  leading from a reachable marking $M_{base}$ to the same, i.e.
  $M_{base}\goesto[\{t_s\}]\cdots\goesto[\{t_e\}]M_{base}$.
  }
  Then every $t_i$ produced tokens that were dependent
  on the same labels as the tokens on its preplaces.
\end{lemma}
\begin{proof}
  Assume the opposite, i.e. there is a $t_i$ for $s \leq i \leq e$ such that $t_i$
  consumed an $L$-independent token from one of its preplaces (for some $L \subseteq \Act$),
  but produced no $L$-independent tokens.
  This $L$-independent token needs to be replaced to again reach $M_{base}$.
  However the replacement token needs to be $L$-independent as otherwise a dependency marking
  different from $M_{base}$ would be reached.
  This token can thus not depend on any of the tokens produced by $t_i$, as it
  would then not be $L$-independent.
  In other words, had $t_i$ not fired, a new $L$-independent
  token could also have been produced on its preplaces, i.e. $N$ would not be
  1-safe, violating the assumptions. Hence no such $t_i$ can be fired,
  or equivalently, every $t_i$ produced tokens that were dependent
  on the same labels as the tokens on its preplaces (which hence all have the
  same dependencies).
\end{proof}

\noindent
We will now show that, given an arbitrary finite, 1-safe net, it is not possible in general to
find a finite, 1-safe, and distributed net which is completed pomset trace
equivalent to the original. As a counterexample, consider the repeated pure \textbf{M} of Figure
\ref{fig-counterexample}. It is a simple net allowing to perform several
transitions of $a$ and $c$ in parallel, and terminating with a single
transition $b$. The main argument of the following proof proceeds as follows: To perform an
arbitrary number of $a$ and $c$-transitions within a finite net there has to be a
loop. To terminate with $b$ the process has to escape from that loop by
disabling all transitions leading to $a$ or $c$. Therefore either a single token is
consumed that is dependent on $a$ as well as on $c$, or two different tokens -- one
$a$-dependent and one $c$-dependent -- are consumed. In the first case an
additional iteration of the loop results in an additional causal dependency,
i.e., in a causal dependency between $a$ and $c$. In the second case the net is
not distributed in the sense of \refdf{distributed}.

\begin{theorem}{main}{}
  It is in general impossible to find for a finite, 1-safe net a
  distributed, completed pomset trace equivalent, finite, 1-safe net.
\end{theorem}
\begin{proof}
  Via the counterexample given in \reffig{counterexample}.
  Suppose a finite, 1-safe, distributed net $N_{impl}$, which is completed
  pomset trace equivalent to the net of \reffig{counterexample}, would exist.
  By refining every $b$-labelled transition in $N_{impl}$ into two transitions
  in the manner of \reflem{implementationrefinement},
  a new net $N = (S, T, F, M_0, \ell)$ is derived.
  By \reflem{implementationrefinement} this new net is finite, 1-safe, distributed and
  completed pomset trace equivalent to the net in \reffig{counterexample} since $N_{impl}$ is.

  $N$ has $|S|$ places and $3$ different labels,
  every place can hold either no token, or a token dependent on any possible combination
  of the three labels. Since $N$ is finite so is $|S|$.
  Hence $N$ has at most $9^{|S|}$ reachable dependency markings. Let $m := 9^{|S|}$.
  $N$ is able to fire $(ac)^mb$ without any step containing more than a single transition
  since the net of \reffig{counterexample} is and the two are assumed to be completed
  pomset trace equivalent.
  Let $G_1, G_2, \ldots G_n$ be the steps fired while doing so. $|G_i| = 1$ for all $i$.
  In the course of firing that sequence, at least one dependency marking is bound to be reached twice.
  Of all those dependency markings which occur twice, we take the one occurring last while
  firing $(ac)^mb$ and call it $M_{base}$.
  Let $G_s, G_{s+1}, \ldots, G_{e-1}, G_e$ be a sequence of steps between two
  occurrences of $M_{base}$, i.e.
  $M_0 \times \{\varnothing\} \goesto[G_1]\goesto[G_2]\cdots M_{base}\goesto[G_s]\cdots\goesto[G_e]M_{base}\cdots\goesto[G_n]$.

  Using \reflem{dependenceinloop} the transitions of the steps $G_s$ to $G_e$ can 
  be partitioned into subsets $T_X$ based
  on the dependencies of the tokens they produced and consumed.
  A set $T_X$ includes all transitions producing
  $X$-dependent, $\Act\setminus X$-independent tokens.
  By firing $G_s\cap T_{\{a\}}, G_{s+1}\cap T_{\{a\}}, \ldots, G_e\cap T_{\{a\}}$
  (skipping empty steps) repeatedly, $M_{base} \Goesto[a^m]$.
  By firing $G_s\cap T_{\{c\}}, G_{s+1}\cap T_{\{c\}}, \ldots, G_e\cap T_{\{c\}}$
  (skipping empty steps) repeatedly, $M_{base} \Goesto[c^m]$.

  We now search for the marking, where the decision to fire $b$ is made.

  Assume a reachable marking $M''$ of $N$ with $M'' \Goesto[a^m]$. If $M'' \not\Goesto[c^m]$ this holds for all $M'''$ reachable from $M''$ since $c$ cannot be enabled using
  tokens produced by a transition labelled $a$ or $b$.
  Otherwise there would exist a pomsets of $N$ in which a $c$ is causally dependent on an $a$ or $b$. Such a
  pomset however does not exist for the net of \reffig{counterexample} thereby violating the
  assumption of completed pomset trace equivalence. 
  If however $c$ is not re-enabled after $M''$ a maximal process including
  finitely many $c$ but infinitely many $a$'s can be produced also leading to a pomset not present
  in the net of \reffig{counterexample}.
  The same argument can be applied with the r\^oles of $a$ and $c$ reversed, hence
  $M'' \Goesto[a^m]$ iff $M'' \Goesto[c^m]$.

  We start from $M_{base}$ and start to fire the steps $G_s$, $G_{s+1}, \ldots, G_n$
  until $a^m$ cannot be fired any more for the first time. This step always
  exists as after $b$ no further $a$'s or $c$'s may be fired.
  Call the single transition in that step $t_b$.
  The marking right before that transition fired, we call $M$, the one right after it $M'$.
  Not only $M \Goesto[a^m]$ but also $M \Goesto[c^m]$ and not only $M' \not\Goesto[a^m]$ but also
  $M' \not\Goesto[c^m]$, as both $M$ and $M'$ are reachable markings.

  $t_b$ is not itself labelled $b$,
  as the refined net has a $\tau$-transition before the $b$, and once a token
  resides on the intermediate place, no $a$-transitions can be fired any more,
  as otherwise a pomset where an $a$ which is not a causal predecessor to a $b$ would be produced,
  again not existing for the net of \reffig{counterexample}.

  To disable the trace $a^m$, the transition $t_b$ needed to consume a token. If $t_b$ had not fired,
  some $G_i\cap T_{\{a\}}$, $s \leq i \leq e$ could have consumed that token,
  hence that token must be $a$-dependent, $c$-independent.
  Similarly, $t_b$ must have consumed a token which could have led to $c^m$.
  This token needs to be $c$-dependent, $a$-independent.
  Hence $t_b$ has at least two preplaces, which in turn are also preplaces to two different
  transitions, call them $t_a$ and $t_c$, which then lead to $a^m$ and $c^m$ respectively.%
  \footnote{The removal of the token leading to $a^m$ and the one leading to $c^m$ must indeed be done
  by a single transition $t_b$ as only a single transition was fired between $M$ and $M'$ and both
  traces were possible in $M$ but impossible in $M'$.}
  As they have common preplaces $t_a$, $t_b$ and $t_c$ are on the same location.

  From $M$ the net can fire $a^m$ consuming only $a$-dependent,
  $c$-independent tokens. It can also fire $c^m$ consuming only $c$-dependent, $a$-independent tokens. 

  Hence there is a sequence of steps leading from $M$ to a marking where $t_a$ is enabled, yet
  only $a$-dependent, $c$-independent tokens have been removed or added.
  Similarly there is a firing sequence leading from $M$ to a marking where $t_c$ is enabled, yet
  only $c$-dependent, $a$-independent tokens have been removed or added.
  As they change disjunct sets of tokens, these two firing sequences can be concatenated, thereby
  leading to a marking where $t_a$ and $t_c$ are concurrently enabled, yet they are on the same location,
  thereby violating the implementation requirements.
\end{proof}

\noindent
Note that the self-loops of the counterexample are not critical to the success of the proof.

This paper only considered 1-safe nets as possible implementations. We conjecture however, that the
proof of \refthm{main} can be extended to non-safe nets as well, as from a place where tokens of different
dependency mix, a transition can always choose the most-dependent token.
In particular a transition intended to produce independent tokens cannot have such a place as a preplace.
Hence every part of the net providing
independent tokens can do so without depending on firings of labelled transitions.
The number of independent tokens produced on a place where a labelled transition consumes
them is thus either finite over every run of the system, or unbounded even without any labelled
transition ever firing. In both cases that place is unsuitable for disabling a potentially
infinitely often occurring loop. If only finitely many tokens are produced, the loop can no longer
happen infinitely often, if an unbounded number of tokens can be produced, no disabling can
be guaranteed.

\section{Conclusion}\label{sec-conclusion}

A review of existing literature in the related area can be found in
\cite{glabbeek08syncasyncinteraction}, nonetheless we wish to refer the reader also to
\cite{hopkins91distnets}, where instead of requiring the equivalence between
specification and implementation to
preserve parallelism, more structural resemblance of the implementation to the
specification is required.

A paper not covered earlier is \cite{badouel02distributing}, where an algorithm for the
automated synthesis of distributed implementations of protocols is presented.
The notion of distributed Petri nets employed therein differs from ours by not
requiring formally that no parallelism may occur on the same location. The authors
however finally generate a finite automaton for each location, again serialising all
actions on a single location. In contrast to the present paper and similar to
\cite{hopkins91distnets}, the authors start with a user-supplied map from
events to locations, and answer the concrete problem of whether that specific
distribution is realisable or not instead of requiring the maximal possible
parallelism to be realised.

Comparing the proof of \refthm{main} with the proof in
\cite{glabbeek08syncasyncinteraction} we observe that the counterexample in
both proofs is based on two conflicts overlapping by a transition, i.e., on what is
therein called a
fully reachable pure \textbf{M}. In the synchronous setting such an overlapping
conflict is solved by the simultaneous removal of tokens on different places in
the preset. In an asynchronous setting these two conflicts have to be
distributed over at least two locations. Intuitively, the problem with such a
distribution is that it prevents the simultaneously solution of the original
overlapping conflicts. Instead these two conflicts have to be solved
in some order. This order must, as done within the encoding presented in
\cite{schicke09synchrony}, be enforced by the encoding,
leading to additional causal dependencies.

The present paper adds another patch to the emerging map of the separation plane
between those equivalences from the spectrum of behavioural equivalences which allow
asynchronous implementation in general and those which do not.
In \cite{glabbeek08syncasyncinteraction} we showed that Petri nets cannot in general be
implemented up to step readiness equivalence, thereby giving an upper bound for distributability
along the branching-time dimension. The present paper provided an upper bound on
the dimension of causality.
We did not formally proof that this bound is tight, and one might imagine that a
behavioural equivalence closer to the notion of dependency markings exists. However,
we were unable to find an equivalence which is sensitive to the local deadlock
problem outlined in \reffig{deadlock} and is not based on processes.
The implementation of \cite{schicke09synchrony} can
serve as a lower bound on both dimensions.
It would be interesting to answer
the implementability question for systems which feature real-valued time, thereby
enabling timeout detection and simultaneous action without co-locality.

That the observed effects are not peculiarities of the Petri net model of systems but
a reality of asynchronous systems in general is underlined by the existence of an
companion paper \cite{peters11asynchronouspi}, giving a result similar to the one
achieved here in the setting of the asynchronous $\pi$-calculus.

A closer look on the proof in \cite{peters11asynchronouspi} reveals that this
proof depends on counterexamples that are so called symmetric networks
including mixed choices in a similar way as our result depends on counterexamples including a
pure \textbf{M}. A symmetric network -- for instance $
R = \overline{a} + b + b.\checkmark \mid \overline{b} + a + a.\checkmark $ in
the second part of the proof -- consists of some
parallel processes that differ only due to some permutation of names. In
combination with mixed choice, i.e., a choice between input as well as output
capabilities, symmetric networks result in conflicting steps on different
links. Hence in both cases the counterexamples refer to some situation in the
synchronous setting in which there are two distinct but conflicting steps. To
solve this conflict two simultaneous activities are necessary -- in case of Petri
nets two tokens are removed simultaneously and in case of the $ \pi $-calculus
two sums are reduced simultaneously in one step. In the asynchronous setting
this simultaneous solution has to be serialised by some kind of lock. It
blocks the enabling of the asynchronous implementations of source steps, such
that no two implementations of conflicting source steps are enabled
concurrently. In both formalisms, Petri nets and the $ \pi $-calculus, it is
this temporally blocking of the implementation of source steps, necessary to
avoid deadlock or divergence in case of conflicting source steps, that leads to
additional causal dependencies.

Apart from this apparent similarity however, much of the relation between the two
results remains mysterious to us. To begin with, the requirements imposed on
Petri net implementations and $\pi$-calculus implementations take wildly different
forms. Additionally, in contrast to the $\pi$-calculus result, the present paper
connected implementation and original by means of behaviour only without any
reference to the system structure. The $\pi$-calculus result on the other hand
had no need to give special
attention to infinite implementations. Finally, we also have no explanation for why
the difference in expressive power (the $\pi$-calculus is turing-complete) should
not make a difference for results such as this. We hope to answer some of these
questions in future work.

The question up to which behavioural equivalence \emph{general} Petri nets are
implementable can also be reversed into the question what properties or substructures of a Petri net
make it unimplementable. One problematic structure for causal equivalences,
identified in this paper, is the net of \reffig{counterexample}, possibly with a
more elaborate route from $a$ and $c$ back to the marking enabling all three transitions.
We did not prove that no fundamentally different problematic structures exists,
but we conjecture that this is indeed the case.

\bibliographystyle{eptcs}
\bibliography{petri}

\end{document}